\documentclass[a4paper,UKenglish,cleveref,autoref, thm-restate]{lipics-v2021}

\hideLIPIcs  

\title{Differentially Private High-Dimensional Approximate Range Counting, Revisited} 

\author{Martin Aumüller}{IT University of Copenhagen, Denmark}{maau@itu.dk}{https://orcid.org/0000-0002-7212-6476}{}

\author{Fabrizio Boninsegna\footnote{Corresponding author, work partially done while visiting IT University of Copenhagen.}}{Department of Information Engineering, University of Padova, Italy}{fabrizio.boninsegna@phd.unipd.it}{https://orcid.org/0009-0001-0996-010X}{}

\author{Francesco Silvestri}{Department of Information Engineering, University of Padova, Italy}{silvestri@dei.unipd.it}{https://orcid.org/0000-0002-9077-9921}{}

\authorrunning{M. Aumüller, F. Boninsegna and F. Silvestri} 

\Copyright{Martin Aumüller, Fabrizio Boninsegna and Francesco Silvestri} 



\category{} 

\relatedversion{} 

\funding{This work was supported in part by
the MUR PRIN 20174LF3T8 AHeAD project, by MUR PNRR CN00000013 National Center for HPC, Big Data and
Quantum Computing, and by Marsden Fund (MFP-UOA2226).}

\acknowledgements{The authors would like to thank Ninh Pham and Rasmus Pagh for useful discussions.}

\nolinenumbers 

\EventEditors{Mark Bun}
\EventNoEds{1}
\EventLongTitle{6th Symposium on Foundations of Responsible Computing (FORC 2025)}
\EventShortTitle{FORC 2025}
\EventAcronym{FORC}
\EventYear{2025}
\EventDate{June 4--6, 2025}
\EventLocation{Stanford University, CA, USA}
\EventLogo{}
\SeriesVolume{329}
\ArticleNo{15}

\usepackage[textsize=small]{todonotes}
\usepackage{mathtools}

\def\Pr{\textnormal{Pr}}

\newcommand{\bff}[1]{{\bf{#1}}}
\newcommand{\inner}[2]{\langle #1\,,\, #2 \rangle}

\newcommand{\BO}[1]{O\left(#1\right)}

\usepackage{algorithm}
\usepackage[noend]{algpseudocode}
\begin{document}

\maketitle

\begin{abstract}
Locality Sensitive Filters are known for offering a quasi-linear space data structure with rigorous guarantees for the Approximate Near Neighbor search (ANN) problem. Building on Locality Sensitive Filters, we derive a simple data structure for the Approximate Near Neighbor Counting (ANNC) problem under differential privacy (DP). Moreover, we provide a simple analysis leveraging a connection with concomitant statistics and extreme value theory.
Our approach produces a simple data structure with a tunable parameter that regulates a trade-off between space-time and utility. Through this trade-off, our data structure achieves the same performance as the recent findings of Andoni et al. (NeurIPS 2023) while offering better utility at the cost of higher space and query time.
In addition, we provide a more efficient algorithm under pure $\varepsilon$-DP and elucidate the connection between ANN and differentially private ANNC.
As a side result, the paper provides a more compact description and analysis of Locality Sensitive Filters for Fair Near Neighbor Search, improving a previous result in Aum\"{u}ller et al. (TODS 2022).
\end{abstract}

\section{Introduction}
Since the emergence of deep learning-based text and image embeddings, such as CLIP~\cite{DBLP:conf/icml/RadfordKHRGASAM21}, the management of collections of high-dimensional vectors has become a critical challenge.
Efficiently handling these collections and supporting complex query operations is essential for various applications, including social networks~\cite{bayardo2007scaling} and recommendation systems~\cite{sarwar2001item}.
The required query operations in applications are often similarity search primitives, which have been widely studied in the literature~\cite{surveyNN, WangSSJ14}.
In particular, the \textit{$r$-Near Neighbor Search} ($r$-NNS) or \emph{Count} ($r$-NNC) problems are fundamental primitives:
given a set $\mathcal{S}\subset \mathbb{R}^d$ of $n$ vectors of $d$ dimensions and a radius $r>0$, construct a data structure that, for any query $\bff{q}\in \mathbb{R}^d$, \emph{returns a point} in $\mathcal{S}$ with distance at most $r$ from $\bff{q}$ if such a point exists, or \emph{counts the number of such points}.
Unfortunately, these problems suffer from the \textit{curse of dimensionality}, which refers to the phenomenon that any exact data structure with polynomial size requires a query time exponential in the dimension of the input space. 
This is supported by popular algorithmic hardness conjectures \cite{AlmanW15,WILLIAMS2005357}.
To address this issue, approximate approaches have been proposed: for a given approximation factor $c>1$, we consider the \textit{$(c,r)$-Approximate Near Neighbor Search} $(c,r)$-ANNS problem and the \emph{$(c,r)$-Approximate Near Neighbor Count} $(c,r)$-ANNC problem. 
These relax the original problem constraints in such a way that the data structure may use points at distance at most $cr$ to answer a query. For a search operation, this means that a point at distance at most $cr$ from $\bff{q}$ can be returned; for a count operation, points at distance between $r$ and $cr$ may be counted as near neighbors.
Locality Sensitive Hashing (LSH)~\cite{indyk1998approximate} and Locality Sensitive Filters (LSF)~\cite{andoni2017optimal,christiani2017framework} are the most common approaches for solving approximate near neighbor problems with theoretical guarantees.

In another line of research, there is an increasing demand for developing solutions for data analysis tasks that preserve the privacy of sensitive personal information in the input data. This was, for example, highlighted in an invited talk at PODS in 2019 by Dwork~\cite{DBLP:conf/pods/Dwork19}, who discussed the application of \emph{differential privacy} (DP)~\cite{dwork2006calibrating} in the US census.
Differential privacy is a widely adopted notion that can provide rigorous guarantees of privacy:
intuitively, DP guarantees that the removal or addition of a single input entry cannot significantly affect the final result of an analysis task. The main idea of the DP approach is to inject well-calibrated noise into the data (structure) that protects privacy without significantly affecting the accuracy of the analysis task. 
Counting queries, which can be seen as a more general form of the near neighbor counting problem studied in this paper, are a well-studied area in DP (see \cite{hardt2010multiplicative, vadhan2017complexity}). However, the error in the final count must be polynomial in $n$ and the dimension of the input space to guarantee privacy \cite{hardt2010geometry, muthukrishnan2012optimal}. To reduce the error, the requirements should be relaxed to provide any count within a fuzzy range around the query \cite{huang2021approximate}. Although this is different from the curse of dimensionality above, achieving an efficient solution similarly necessitates the use of approximate counts.
Thus, it is \emph{natural} to study approximate counting problems for the high-dimensional ANN problem when considering them in a DP setting.

In this paper, we focus on ANN problems for the inner product on the  $d$-dimensional unit sphere $\mathbb{S}^{d-1} \coloneqq   \{\bff{x} \in \mathbb{R}^d: \| \bff{x} \|_2 = 1\}$, which is often called \emph{cosine} or \emph{angular similarity}. In these settings, the goal is to \emph{find} or \emph{count} points with large inner products. 
More specifically, we study the $(\alpha,\beta)$-ANN count and search problems with $0\leq \beta < \alpha < 1$. Let $B(\bff{q}, \alpha) \coloneqq \{\bff{x} \in \mathbb{S}^{d-1} \mid \inner{\bff{x}}{\bff{q}} \geq \alpha\}$ be the set of unit vectors that have an inner product of at least $\alpha$ with $\bff{q} \in \mathbb{S}^{d-1}$. The counting variant asks for a query $\bff{q}$ to count all points in a dataset $\mathcal{S}$ with inner product at least $\alpha$ but tolerates points with inner product at least $\beta$. That means that the resulting estimate $E$ should satisfies $|\mathcal{S}\cap B(\bff{q},\alpha)| \leq E \leq |\mathcal{S}\cap B(\bff{q},\beta)|$. 
This is a common notation in inner product search, and intuitively $\alpha$ and $\beta$ are equivalent to $r$ and $r/c$ in $(c,r)$-ANNC.
The first result for differentially private $(\alpha,\beta)$-ANNC in high dimensions has been recently provided by Andoni et al. \cite{andoni2024differentially}, where the authors use a linear space tree-based data structure based on the concept of LSF.%
\footnote{As observed by \cite{andoni2024differentially} as well, a solution on the unit sphere leads to a solution for the whole Euclidean space thanks to embedding methods. We will defer all discussion of this embedding and its applicability to Appendix~\ref{appendix: Embedding into the Euclidean sphere}.}

In this work, we explore the design space of locality sensitive filtering-based solutions to ANN and ANNC problems.
We show that by revisiting the LSF framework for ANN it is possible to derive a simpler and more compact solution for ANN.
Building on this result, we derive a novel solution for ANNC under differential privacy that extends the range of applicability of the state of the art~\cite{andoni2024differentially} by removing some limitations on parameter ranges and differential privacy assumptions.
In particular, we provide strong guarantees in the regime of \emph{pure DP} (in contrast to approximate DP), and show that balancing the noise term of DP with the approximation error of LSF is not the only design choice: in fact, spending more space and query time results in more accurate solutions.
The following section will provide more details on the technical contribution.

\subsection{Our Contribution}
\label{section: Our contributions}
\begin{algorithm}[t]
\footnotesize
    \centering
    \caption{\texttt{DPTop-1 Data Structure for DP-ANNC}}\label{alg: DPTop-1 Data Structure}
    \begin{minipage}[t]{0.52\textwidth}
    \begin{algorithmic}[1]
        \Procedure {\texttt{construction}$ (\mathcal{S}\subseteq\mathbb{S}^{d-1}, \alpha, \beta, \varepsilon, \delta)$}{}
        \State $\rho \gets \frac{(1-\alpha^2)(1-\beta^2)}{(1-\alpha\beta)^2}$; 
         $m\gets \big\lceil n^{\frac{\rho}{1-\alpha^2}} \big\rceil$
        \State $\mathcal{A}^{m} = (\bff{a}_1, \dots, \bff{a}_m)$ with $\bff{a}_i \sim \mathcal{N}(0,1)^d$
        \State $T[1,\ldots,m] \gets (0, \ldots, 0)$
        \For {$\bff{x} \in \mathcal{S}$}
            \State $i\gets \arg\max_{i \in [m]}\inner{\bff{a}_i}{\bff{x}}$
            \State $T[i] \gets T[i] + 1$
        \EndFor
        \State $\eta \gets \alpha\sqrt{2\log m}-\sqrt{2(1-\alpha^2)\log\log m}$
        \State $\tilde{T} \gets \texttt{make\_private}(T, \varepsilon, \delta)$
        \State \textbf{return }$\mathcal{D} = ({T}, \mathcal{A}^m, m,  \eta)$
        \EndProcedure
        \end{algorithmic}
        \end{minipage}
        \hfill
        \footnotesize
        \begin{minipage}[t]{0.46\textwidth}
        \begin{algorithmic}[1]
        \Procedure{$\texttt{query}(\bff{q})$}{}
            \State $B\gets \{i \in [m]: \inner{\bff{a}_i}{\bff{q}}\geq \eta\}$
            \State $\widetilde{\text{ans}}\gets 0$
            \For {$i \in B$}
                \State $\widetilde{\text{ans}}\gets \widetilde{\text{ans}}+{\tilde{T}}[i]$
            \EndFor
            \State \textbf{return }$\widetilde{\text{ans}}$
        \EndProcedure
        \end{algorithmic}
    \end{minipage}
\end{algorithm}
\paragraph*{Revisiting LSF for ANN}
Our work is based on a construction for $(\alpha, \beta)$-ANN  first proposed by Aumüller et al. \cite{aumuller2022sampling} in the context of algorithmic fairness.
We provide a more compact description and analysis of LSF for $(\alpha, \beta)$-ANN, and obtain a data structure with a lower pre-processing time of $O(d\,n^{1 + o(1)})$ (from $O(d\,n^{1 + \rho + o(1)})$ where $0 < \rho :=\rho(\alpha, \beta) \leq 1$ is the \emph{strength of the filter}).
Moreover, assuming that some random variables follow a limiting distribution (Theorem \ref{theorem: asymptotic concominants}), we get a more compact and simpler proof than~\cite{aumuller2022sampling} leveraging an elegant connection with concomitant statistics and extreme value theory, and the parameters used in our solution naturally follow from this theory. In Section~\ref{subsection: analysis of CloseTop-1}, we demonstrate that an alternative construction procedure, \texttt{CloseTop-1}, remains effective without assuming any limiting distributions, matching the convergent properties.

\paragraph*{From ANN to DP-ANNC}
We then present a solution for $(\alpha, \beta)$-ANNC under differential privacy.
More specifically, we provide a general methodology that allows to translate a variant of a list-of-points data structure \cite{andoni2017optimal} for ANN into a data structure for DP-ANNC. Intuitively, a list-of-points is a data structure where input points are organized in a collection of lists and a query consists of a scan of some of these lists: this is the case, for instance, of methods based on LSH or LSF. This approach offers a way to develop the data structure in two steps by describing a data structure satisfying certain characteristics for ANN, and then applying a suitable DP mechanism on top of it.

When the data structure is built on top of the previous result, we get the \texttt{DPTop-1} data structure presented in Algorithm~\ref{alg: DPTop-1 Data Structure}, which we will now describe in words. Given a dataset $\mathcal{S} \subseteq \mathbb{S}^{d-1}$ consisting of $n$ points, two similarity thresholds $0 \leq \beta < \alpha < 1$, and privacy parameters $\varepsilon>0$ and $\delta\in [0,1)$, the data structure samples $m = n^{O(1)}$ Gaussian vectors. We associate with each such vector a counter, initialized with 0; each point in $\mathcal{S}$ increments the counter of the vector that maximizes the inner product. 
Then, the counts are made differentially private by a suitable DP mechanism \texttt{make\_private}, for example the Truncated Laplace mechanism~\cite{geng2020tight} or the ALP mechanism~\cite{aumuller2022representing}. Depending on the mechanism used, different privacy guarantees can be provided.
Since each point increments exactly one counter, the absence or presence of a data point affects only a small part of the data structure. As a result, the sensitivity of the data structure is low and only a small amount of noise has to be added.
The $m$ vectors with their noisy counts form the data structure that can be released publicly. 
For any query $\bff{q}\in \mathbb{S}^{d-1}$, the estimate of the number of near neighbors with inner product at least $\alpha$ is the sum of counters associated to vectors with inner product similarity to $\bff{q}$ greater than $\eta(\alpha) = \alpha\sqrt{2\log m}-\sqrt{2(1-\alpha^2)\log\log m}$. This choice is guided by the theory of concomitant statistics and extreme value theory of Gaussian random variables~\cite{david1974asymptotic}, which we will formally introduce in Section~\ref{subsection: Concominant Order Statistics}. As we will detail in Section~\ref{section: Top-1}, in the asymptotic regime---hence for $n\to\infty$--- \texttt{DPTop-1} offers a simple and elegant solution for the $(\alpha,\beta)$-ANNC problem under differential privacy. 
The following theorem provides the guarantees of \texttt{DPTop-1} when using the Truncated Laplace mechanism~\cite{geng2020tight} as a privacy mechanism (we refer to Theorem~\ref{theorem: Top-1} for the exact statements regarding the ANN data structure, and Theorem \ref{theorem: DP-ANNC} for the exact statements of the DP-ANNC implementation).
\begin{theorem} 
\label{theorem: our contribution DP-ANNC} Consider the asymptotic regime, $\varepsilon > 0$, $\delta \in \big(0, \frac{1}{2}\big)$, $0\leq\beta<\alpha<1$, and $\alpha - \beta=\Omega\big(\sqrt{\frac{\log\log n}{\log n}}\big)$. Let $\mathcal{S}=\{x_i\}_{i=1,\dots, n} \subseteq \mathbb{S}^{d-1}$ and let $\bff{q} \in \mathbb{S}^{d-1}$.
Then \texttt{DPTop-1} (Algorithm \ref{alg: DPTop-1 Data Structure}) with Truncated Laplace mechanism  
satisfies $(\varepsilon, \delta)$-DP, and with probability at least 2/3,  the query 
returns $\widetilde{\textnormal{ans}}$ such that
\begin{equation*}
    (1-o(1))|\mathcal{S}\cap B(\bff{q},\alpha)|-O\bigg(\frac{\log(1/\delta)}{\varepsilon}n^{\rho + o(1)}\bigg)\leq \widetilde{\textnormal{ans}}\leq |\mathcal{S}\cap B(\bff{q}, \beta)|+O\bigg(\frac{\log(1/\delta)}{\varepsilon}n^{\rho + o(1)}\bigg).
\end{equation*}
The data structure has pre-processing time $O(d\cdot n^{1+\frac{\rho}{1-\alpha^2}})$, expected query time and space $O(d\cdot n^{\frac{\rho}{1-\alpha^2}})$.
\end{theorem}
This simple algorithm matches the accuracy of the solution found by Andoni et al.~\cite{andoni2024differentially} and results in a straightforward space partitioning of $\mathbb{S}^{d-1}$, one of the main goals of~\cite{andoni2024differentially}.
Furthermore, our approach provides a solution that works for almost all similarity thresholds on the unit sphere, while~\cite{andoni2024differentially} supports a single distance threshold and relies on embedding techniques and scaling for all other distance thresholds.

While Algorithm \ref{alg: DPTop-1 Data Structure} is potentially already practical, both the space and running time requirements are worse than the solution presented in~\cite{andoni2024differentially} due to the large number of filters $m$.
We suggest two improvements to the algorithm that do not affect the compactness of the algorithm and the proof.
We first observe that the theoretical result in Theorem \ref{theorem: our contribution DP-ANNC} works only for $n\to \infty$: we drop this limitation in Section~\ref{subsection: analysis of CloseTop-1} thanks to a small but novel variation in the construction procedure.
We then observe in Section~\ref{subsection: tensorization} how to achieve almost linear pre-processing time, linear space, and $d\cdot n^{\rho + o(1)}$ expected query time, which is optimal due to a space-time tradeoffs lower bound \cite{andoni2017optimal}\footnote{It is sufficient to set $\rho_u=0$ in Theorem 3.3 in the reference paper to get the lower bound for the running time}, by concatenating $\text{polylog}(n)$ data structures, using a technique called \emph{tensorization} \cite{christiani2017framework, aumuller2022sampling}. 
We call this final version \texttt{TensorCloseTop-1}.
Table~\ref{tab:my_label} summarizes the guarantees of all algorithms described in this paragraph and compares them with the state of the art approach~\cite{andoni2024differentially}.
\begin{table}
\footnotesize
\centering
\renewcommand{\arraystretch}{3} 
\setlength{\tabcolsep}{3pt} 
\caption{\textbf{Results for DP-ANNC},  $\sigma = \frac{(1-\alpha^2)(1-\beta^2)}{(1-\alpha\beta)^2+(\alpha-\beta)^2}$, $\rho = \frac{(1-\alpha^2)(1-\beta^2)}{(1-\alpha\beta)^2}$, and $\sigma<\rho<2\sigma$ for any $0\leq\beta<\alpha<1$. The bound $(*)$ holds only in expectation. Time and space bounds omit factor $d$.} 
\begin{tabular}{|>{\centering\arraybackslash}m{3.2cm}|
>{\centering\arraybackslash}m{1.2cm}|
>{\centering\arraybackslash}m{2.5cm}|
>{\centering\arraybackslash}m{2.3cm}|
>{\centering\arraybackslash}m{1.8cm}|
>{\centering\arraybackslash}m{1.7cm}|}
     \hline
     
     \vspace{-6pt}\textbf{Mechanism} & 
     \vspace{-6pt}\textbf{Privacy} & 
     \vspace{-6pt}\textbf{Additive Error} & 
     \vspace{6pt}\shortstack{\textbf{Preprocessing} \\ \textbf{Time}}&
     \vspace{6pt}\shortstack{\textbf{Expected} \\ 
     \textbf{Query} \\
     \textbf{Time}} & 
     \vspace{-6pt}\textbf{Space} \\
     
     \hline
     
     \vspace{6pt}\shortstack{\texttt{DPTop-1} \\ 
     w/ Truncated Laplace} & 
     \vspace{-6pt}$(\varepsilon, \delta)$ & 
     \vspace{-6pt}$O\big(\frac{\log(1/\delta)}{\varepsilon}n^{\rho + o(1)}\big)$ & 
     \vspace{-6pt}$O\big(n^{1+\frac{\rho}{1-\alpha^2}}\big)$ & 
     \vspace{-6pt}$O\big(n^{\frac{\rho}{1-\alpha^2}}\big)$ & 
     \vspace{-6pt}$O\big(n^{\frac{\rho}{1-\alpha^2}}\big)$ \\
     
     \hline
     
     \vspace{6pt}\shortstack{{Andoni et al.~\cite{andoni2024differentially}}, and \\
    \texttt{TensorCloseTop-1} \\ 
    w/ Truncated Laplace} & 
     \vspace{-6pt}$(\varepsilon, \delta)$ & 
     \vspace{-6pt}$O\big(\frac{\log(1/\delta)}{\varepsilon}n^{\rho + o(1)}\big)$ & 
     \vspace{-6pt}$n^{1+o(1)}$ & 
     \vspace{-6pt}$n^{\rho + o(1)}$ & $O(n)$ \\
     
     \hline
     \hline
     
     \vspace{-6pt}\shortstack{\footnotesize {Andoni et al.~\cite{andoni2024differentially}}}  & 
     \vspace{-6pt}$(\varepsilon, 0)$ & 
     \vspace{-6pt}$O\big(\frac{1}{\varepsilon}n^{\rho + o(1)}\big)$ & 
     \vspace{-6pt}$n^{1+o(1)}$$^{(\ast)}$ & 
     \vspace{-6pt}$n^{\rho + o(1)}$ & $O(n)^{(\ast)}$ \\
     
     \hline
     
     \vspace{6pt}\shortstack{\texttt{TensorCloseTop-1} \\ w/ \texttt{Max Projection}} &  
     \vspace{-6pt}$(\varepsilon, 0)$ & 
     \vspace{-6pt}$O\big(\frac{1}{\varepsilon}n^{\rho + o(1)}\big)$ & 
     \vspace{-6pt}$n^{1+o(1)}$ & 
     \vspace{-6pt}$n^{\rho + o(1)}$ & $O(n)$ \\
     
     \hline
     
     \vspace{6pt}\shortstack{Unbalanced \\ \texttt{TensorCloseTop-1} \\ with \texttt{Laplace}} &  
     \vspace{-6pt}$(\varepsilon, 0)$ &  
     \vspace{-6pt}$O\big(\frac{1}{\varepsilon}n^{\sigma + o(1)}\big)$ & 
     \vspace{-6pt}$O\big(n^{\frac{\sigma}{1-\alpha^2}}\big)$ & 
     \vspace{-6pt}$n^{2\sigma + o(1)}$ & $O\big(n^{\frac{\sigma}{1-\alpha^2}}\big)$ \\
     \hline
\end{tabular}
\label{tab:my_label}
\end{table}

\paragraph*{Balanced and Unbalanced Data Structures}
As can be seen from the theorem statement, the provided estimate comes with an additive error of $O\left(\log(1/\delta)n^{\rho+ o(1)}/\varepsilon\right)$.
This term includes two fundamentally different error sources.
First, we might include ``far away points'', i.e., points with inner product below $\beta$, in the count through the ANN data structure. This error does not depend on the choice of the privacy mechanism. 
The second source of error, which is due to privacy, arises from summing over $n^{\rho+o(1)}$ noisy counters, as a query is expected to search through $n^{\rho+o(1)}$ buckets on average. 
\texttt{DPTop-1} and the data structure of Andoni et al. \cite{andoni2024differentially} \emph{balance} these two errors so that both are upper bounded by $n^{\rho+o(1)}$; however, as we will show in Sections~\ref{section: Top-1} and~\ref{section: From ANN to DP-ANNC} this is not necessarily an optimal trade-off. 
By using an \emph{unbalanced} data structure---to be discussed in detail in Section~\ref{subsection: Balanced and Unbalanced Top-1}--- with Laplace noise, we can achieve a more accurate result at the cost of larger running time and more space. 
In fact, for any parameter, our unbalanced data structure solves DP-ANNC with an additive error $n^{\sigma+o(1)}$ for $\sigma<\rho$.
The main insight is that the sum of $n^{\rho+o(1)}$ noisy, unbiased, and uncorrelated counters, provided by the Laplace mechanism, scales with $n^{\rho/2+o(1)}$ by concentration arguments. This makes it suboptimal to provide the same upper bound for both sources of error.

\paragraph*{Comparison to Andoni et al. (NeurIPS 2023)}
For the convenience of the reader,
we now provide a quick description of the solution in \cite{andoni2024differentially}, and highlight differences to the present work.
Let $\eta_u$, $\eta_q$, $T$ and $K$ be suitable parameters that depend on $\alpha$ and $\beta$.
The data structure consists of a tree with degree $T$ and height $K$; each internal node $v$ is associated to a $d$-dimensional random Gaussian vector  $\bff{g}_v$ and to a subset of the input set $\mathcal{S}$.
At the beginning, the root contains the entire input set $\mathcal{S}$.
Then, for each internal node with a top-down approach, we  partition the assigned vectors into $T$ groups: each input vector $\bff{x}$ is assigned to the node $v$ with  smallest index such that the inner product is $\inner{\bff{x}}{\bff{g}_v}\geq \eta_u$. (If no such index exists, the point is not stored in the data structure.)
Once the input points have been processed, we replace the list in each leaf with its noisy size by adding Truncated Laplace noise~\cite{geng2020tight} (if the final value is below a given threshold, we replace the value with 0).
Given a query $\bff{q}$, we collect the counts of all leaves $v$ for which $\inner{\bff{q}}{\bff{g}_{v'}}\geq \eta_q$ for all nodes $v'$ on the unique path from the root to $v$, and return as a final estimate the sum of these counts. Using this tree data structure circumvents the evaluation problem of a too large number of filters mentioned for our variant above. 

As the goal of their paper is to address Euclidean distance, the range of the $\alpha$ and $\beta$ parameters is limited; their analysis works only for an $\alpha$ value that corresponds to Euclidean distance $\Theta(\log^{-1/8} n)$ and all other distances are only supported through embedding and scaling, which adds an additional distortion to the distance values. 
In contrast, our solution allows for a wider range of these parameters, increasing the applicability for inner product similarity on the sphere and still gets a data structure that holds for Euclidean distance.
Furthermore, by removing the tree structure, we are able to design an algorithm with little data dependencies that is likely to exploit hardware accelerators (e.g., Nvidia Tensor Core, Intel AMX) for neural networks that are optimized for batches of inner products.

\subsection{Previous work}
\paragraph*{Near Neighbor Search.}

\textit{Locality Sensitive Hashing} (LSH)~\cite{indyk1998approximate} is one of the most used approaches for solving ANN with rigorous guarantees. However, it suffers of large space requirements. 
Indeed, LSH  requires $\BO{nd + n^{1+\rho}}$ memory words, where $\rho$ is a parameter describing the ``power'' of the used LSH (e.g., $\rho=O(1/c^2)$ for Euclidean distance~\cite{DBLP:journals/cacm/AndoniI08}): indeed the data structure requires to create $n^\rho$ hash tables, each storing all $n$ points. Both Panigrahy~\cite{DBLP:conf/soda/Panigrahy06} and Kapralov~\cite{DBLP:conf/pods/Kapralov15} provided linear space solutions using variants of LSH.
An interesting technique to achieve smooth space-time trade-offs is given by \textit{Locality Sensitive Filters} (LSF)~\cite{christiani2017framework,andoni2017optimal}.
In the context of this work, the interesting space-time trade-off to focus on is the linear space regime~\cite{aumuller2022sampling}. Besides offering optimal space, this regime has many additional interesting properties for downstream applications.
For example, very recently, Andoni et al.~\cite{andoni2024differentially} showed their application in the context of differentially private range counting in high dimensional data. As mentioned above, 
a linear space data structure only involves at most one time a point of the dataset, so the absence or presence of a data point only affects a small part of the data structure; in comparison, with traditional LSH-based approaches a single point is stored in many different tables in the data structure. 

\paragraph*{Differentially Private Counting Queries.}
Counting queries require, except for a few classes of queries, a polynomial error in $n$ and in the space dimension to guarantee privacy  \cite{hardt2010geometry, muthukrishnan2012optimal}.
This incentivized Huang and Yao \cite{huang2021approximate} to relax the condition, allowing for the release of any count within a fuzzy range of the query. For ball queries in $\mathbb{R}^d$, this is essentially the problem to release any count between $|\mathcal{S}\cap B_{D}(\bff{q}, r)|$ and $|\mathcal{S}\cap B_{D}(\bff{q}, cr)|$, that we will identify as the $(c, r)$-Approximate Nearest Neighbor Count (ANNC) problem. One of the main results in \cite{huang2021approximate} is that there exists a differential private solution of the problem with poly-logarithmic error in $n$ at the price of an exponential dependence in the dimension $d$. A solution for the high dimensional case was proposed in \cite{andoni2024differentially}, where Andoni et al. proposed a linear space data structure for the $(c,r)$-ANN, to solve the differential private $(c,r)$-ANNC. The authors developed a Locality Sensitive Filtering data structure with $\rho=\frac{4c^2}{(c^2+1)^2}$, for the differential private $(c, r)$-ANNC in the Euclidean space  with additive error $O\big(n^{\rho +o(1)}\big)$ and multiplicative error $1-o(1)$, getting rid of the dependence on the dimension. The proposed data structure is based on a more general theory for data structures with space-time trade-offs \cite{andoni2017optimal}, making the analysis more involved. In this paper, we will show that our data structure offers the same guarantees with a more streamlined analysis.

\section{Preliminaries}
\label{section: Preliminaries}
\subsection{Notation} 
\label{subsection: notation}
We let $[m]$ be the set of integers $\{1,\dots, m\}$.
We denote with $\bff{q}\in \mathbb{S}^{d-1}$ a query point, and with $\bff{x}_{\varrho}$ a point of the dataset such that $\inner{\bff{x}_{\varrho}}{\bff{q}} = \varrho$. We set $\bff{a}_{\bff{x}}$ as the vector associated to $\bff{x}$, and define $X_{\bff{x}}:=\inner{\bff{a}_{\bff{x}}}{\bff{x}}$, and $Q_{\bff{x}} = \inner{\bff{a}_{\bff{x}}}{\bff{q}}$ as the concomitant---to be defined in Section~\ref{subsection: Concominant Order Statistics}--- of $X_{\bff{x}}$. If $X_{\bff{x}} = \max_{i\in [m]}\inner{\bff{a}_i}{\bff{x}}$ then it is denoted as $X_{\bff{x}, (m)}$ and so the concomitant as $Q_{\bff{x}, [m]}$. The threshold for the query filter is $\eta = \alpha\sqrt{2\log m}-\sqrt{2(1-\alpha^2)\log\log m}$.  A ball in the hyper-sphere under inner product similarity centered in $\bff{q}$ is denoted as $B(\bff{q}, \alpha):=\{\bff{x}\in \mathbb{S}^{d-1}\,:\,\inner{\bff{q}}{\bff{x}}\geq \alpha\}$. We call a point $\bff{x}$ in $\mathcal{S}$ \emph{close} to $\bff{q}$ if  $\inner{\bff{q}}{\bff{x}} \geq \alpha$, and \emph{far} if $\inner{\bff{q}}{\bff{x}} <\beta$. We consider $n$ to be the number of points in the dataset $\mathcal{S}\subset \mathbb{S}^{d-1}$. We denote the Gaussian distribution of mean $\mu$ and variance $\sigma^2$ as $\mathcal{N}(\mu, \sigma^2)$. 

\subsection{Problem Definition}
\label{subsection: problem definition}
\begin{definition}[$(\alpha, \beta)$-ANN] Consider a set $\mathcal{S}\subseteq \mathbb{S}^{d-1}$ of $n$ points. The Approximate Nearest Neighbor Search ANN problem asks to construct a data structure for $\mathcal{S}$ that for a given query $\bff{q}\in \mathbb{S}^{d-1}$, such that $B(\bff{q},\alpha)$ contains a point in $\mathcal{S}$, returns a point in $\mathcal{S} \cap B(\bff{q},\beta)$.
\end{definition}
We will study a data structure that solves this problem with asymptotically high probability, hence at least $1-o(1)$. The inner product similarity is related to the Euclidean distance, as $||\bff{x}-\bff{y}||_2=\sqrt{2(1-\inner{\bff{x}}{\bff{y}})}$ for any $\bff{x}, \bff{y}\in \mathbb{S}^{d-1}$. Therefore, for $\alpha = 1-\frac{r^2}{2}$ and $\beta = 1-\frac{(cr)^2}{2}$, a $(c,r)$-ANN in $(\mathbb{S}^{d-1}, \|\cdot \|_2)$ is equivalent to the $(\alpha,\beta)$-ANN defined above.

\begin{definition}[$(\alpha,\beta)$-ANNC] Consider a set $\mathcal{S}\subseteq \mathbb{S}^{d-1}$ of $n$ points. The Approximate Near Neighbor Counting (ANNC) problem asks to construct a data structure for $\mathcal{S}$ that, for a given query $\bff{q}\in \mathbb{S}^{d-1}$, returns a number between $|\mathcal{S} \cap B(\bff{q}, \alpha)|$ and $|\mathcal{S} \cap B(\bff{q}, \beta)|$.
\end{definition}
This problem is the counting equivalent of the well-studied~\emph{spherical range reporting problem} (see for example~\cite{DBLP:conf/soda/AhleAP17}) that asks to \emph{enumerate} all points at a certain distance from $\bff{q}$. 

\subsection{Concomitant Order Statistics}
\label{subsection: Concominant Order Statistics}
The theory of concomitant order statistics offers a very elegant and intuitive tool for random projections in $\mathbb{S}^{d-1}$, as highlighted in \cite{eshghi2008locality, pham2021simple, pham2022falconn++}.
Let $(X_1, Y_1), \dots, (X_m, Y_m)$ be $m$ random samples from a bivariate distribution. We order the values according to $X$ such that $X_{(1)} \leq \dots \leq X_{(i)} \leq \dots \leq X_{(m)}$.
The $Y$-variate associated with $X_{(r)}$ is denoted as $Y_{[r]}$ and it is called the \emph{concomitant} of the $r$-th order statistic.

\paragraph*{Relation With Random Projections.} 
Let $\bff{x}, \bff{y}\in \mathbb{S}^{d-1}$ such that $\inner{\bff{x}}{\bff{y}}=\varrho$ and $\bff{a} \sim \mathcal{N}(0,1)^{d}$. 
Consider the random variables $X=\inner{\bff{x}}{\bff{a}}$ and $Y = \inner{\bff{y}}{\bff{a}}$, then $(X,Y)\sim  \mathcal{N}(0,0,1,1,\varrho)$, which is a standard bivariate normal distribution with correlation coefficient $\varrho$ \footnote{The general notation for a bivariate Gaussian distribution is $\mathcal{N}(\mu_X, \mu_Y, \text{Var}[X], \text{Var}[Y], \text{Cov}[X,Y])$, while $\text{Cov}[X,Y] = \sum_{i=1}^{d}x_iy_i = \varrho$.}. 
The relation between concomitant and order statistics for the normal bivariate distribution is given by the following lemma.

\begin{lemma}[\cite{david1974asymptotic}] 
\label{lemma: relation between concominant and extreme}
Given $m$ samples $\{(X_i, Y_i)\}_{i=1, \dots, m}$ from the standard bivariate normal distribution $\mathcal{N}(0,0,1,1,\varrho)$, for any $r\in\{1, \dots, m\}$ we have that 
$ Y_{[r]} = \varrho X_{(r)} + Z_r$, where $Z_r$ is a random variable distributed as $\mathcal{N}(0, 1 - \varrho^2)$ and independent of $X_{(r)}$.
\end{lemma}
A standard result of concomitant order statistics states that $Y_{[r]}-\mathbb{E}[Y_{[r]}]$ weakly converges to a Gaussian distribution $\mathcal{N}(0,1-\varrho^2)$ \cite{david1974asymptotic}. Thus, defining $F_{Y_{(m)}}$ as the probability density function of $Y_{[m]}$, we have that $\lim_{m\to \infty}F_{Y_{[m]}} = \mathcal{N}(\varrho\mathbb{E}[X_{(m)}], 1-\varrho^2)$ \cite{fleming2013counting}. By adding the fact that $\mathbb{E}[X_{(m)}] =\sqrt{2\log m} - o(1)$ \cite{hall1979rate} we get the following theorem.

\begin{theorem}[\cite{david1974asymptotic, hall1979rate}]
\label{theorem: asymptotic concominants}
    Let $\{(X_i, Y_i)\}_{i=1,\dots, m}$ be $m$ i.i.d. samples from $\mathcal{N}(0,0,1,1,\varrho)$. Then $Y_{[m]}$ weakly converges to $\mathcal{N}(\varrho \sqrt{2\log m}, 1-\varrho^2)$.
\end{theorem}
This asymptotic result serves as the basis of the intuition for our data structure: if we associate to each point in $\bff{x}\in\mathcal{S}$ the closest Gaussian vector $\bff{a}_{\bff{x}} = \arg\max_{\bff{a} \in \{\bff{a}_1, \dots, \bff{a}_m\}}\inner{\bff{a}}{\bff{x}}$, then a query $\bff{q}\in \mathbb{S}^{d-1}$, such that $\inner{\bff{q}}{\bff{x}}=\varrho$, will find $\bff{x}$ associated to a Gaussian vector with inner product similarity $\inner{\bff{q}}{\bff{a}_{\bff{x}}}\sim \varrho \sqrt{2\log m}$.

\subsection{Differential Privacy}
\label{subsection: differential privacy}
Differential Privacy (DP) is a definition on indistinguishability for the outputs of protocols applied to \emph{neighboring} datasets. Two datasets are neighbor $\mathcal{S}\sim \mathcal{S}'$ by \emph{addition/removal} if they differ by the addition or a removal of one point, instead they are neighbor by \emph{substitution} if $|\mathcal{S}|=|\mathcal{S}'|$ and they differ in one point. 

\begin{definition}[Approximate Differential Privacy \cite{dwork2014algorithmic}]
For $\varepsilon>0$ and $\delta \in [0,1)$, we say that a randomized algorithm $\mathcal{M}$ is $(\varepsilon, \delta)$-differentially private if for any two neighboring datasets $\mathcal{S}\sim \mathcal{S}'$, and any possible outcome of the algorithm $Y \subseteq \textnormal{range}(\mathcal{M})$, we have $\Pr[\mathcal{M}(\mathcal{S})\in Y]\leq e^{\varepsilon}\Pr[\mathcal{M}(\mathcal{S}')\in Y] + \delta$.
\end{definition}
We are mainly interested in \emph{histogram queries}  $f:\mathcal{X}^n\to \mathbb{N}^{|\mathcal{X}|}$ \cite{dwork2006calibrating}, where $\mathcal{X}$ is the data universe and $n$ is the size of the data set. The most common way to privatize $f(\mathcal{S})$ is to obfuscate the true values by adding noise scaled on the \emph{sensitivity} of the query $\Delta_f = \max_{\mathcal{S}\sim  \mathcal{S}'}\|f(\mathcal{S})-f(\mathcal{S}')\|_1$ \cite{dwork2006calibrating}. 
 In our context, each data point contributes to exactly one counter. For the addition/removal neighboring relationship, $\Delta_f = 1$; for substitution, $\Delta_f = 2$.
We consider three different DP mechanisms when privatizing counters, for example, in the function \texttt{make\_private} in Algorithm~\ref{alg: DPTop-1 Data Structure}: Truncated Laplace Mechanism~\cite{geng2020tight}, Laplace Mechanism~\cite{geng2020tight}, and Max Projection~\cite{aumuller2022representing}. More details can be found in Appendix~\ref{appendix: dp}.

\section{\texttt{Top-1} Data Structure for ANN}
\label{section: Top-1}
\begin{algorithm}[t]
\footnotesize
\caption{\texttt{Top-1 Data Structure}}\label{alg: Top-1 Data Structure}
\begin{minipage}[t]{0.52\textwidth}
    \begin{algorithmic}[1]
    \Procedure {\texttt{construction}$(\mathcal{S}\subset \mathbb{S}^{d-1}, \alpha, \beta, \theta)$}{}
    \State $m\gets \big\lceil n^{\frac{\theta}{1-\alpha^2}}\big\rceil$
    \State $\mathcal{A}^{m} = (\bff{a}_1, \dots, \bff{a}_m)$ with $\bff{a}_i \sim \mathcal{N}(0,1)^d$ 
    \State $H \gets \text{empty hash table}$
    \For {$\bff{x} \in \mathcal{S}$}
        \State $i\gets \arg\max_{i \in [m]}\inner{\bff{a}_i}{\bff{x}}$
        \State $H.\texttt{insert}(i, \bff{x})$
    \EndFor
    \State $\eta \gets \alpha\sqrt{2\log m}-\sqrt{2(1-\alpha^2)\log\log m}$ 
    \State \textbf{return }$\mathcal{D} = (H, \mathcal{A}^m, m, \eta)$
    \EndProcedure
    \end{algorithmic}
    \end{minipage}
    \hfill
    \begin{minipage}[t]{0.46\textwidth}
    \begin{algorithmic}[1]
        \Procedure{\texttt{search}$(\bff{q})$}{} 
    \State \textbf{return } $\{i\in [m]\,:\, \inner{\bff{a}_i}{\bff{q}}\geq \eta\}$
\EndProcedure
\newline
\Procedure {\texttt{query}$(\bff{q})$}{}
    \State $B\gets \mathcal{D}.\texttt{search}(\bff{q}, \eta)$
    \For {$i \in B$}
        \For {$\bff{x}\in H[i]$}
            \If {$\inner{\bff{q}}{\bff{x}}\geq \beta$}
                \State \textbf{return } $\bff{x}$
            \EndIf
        \EndFor
    \EndFor
    \State \textbf{return $\perp$} 
\EndProcedure
        \end{algorithmic}
    \end{minipage}
\end{algorithm}
Algorithm~\ref{alg: Top-1 Data Structure} describes the \texttt{Top-1} data structure, which is the variant of Algorithm~\ref{alg: DPTop-1 Data Structure} targeting the $(\alpha, \beta)$-ANN problem.
Let $\mathcal{A}^{m} = (\bff{a}_1, \dots, \bff{a}_m)$ be a set of $m$ random vectors from $\mathcal{N}(0,1)^d$.
The data structure consists of a hash table that stores the input vectors assigned to each random vector in $\mathcal{A}^{m}$: more specifically, we assign each input vector $\bff{x} \in \mathcal{S}$ to the random vector in $\mathcal{A}^{m}$ with the largest inner product.
For a given query vector $\bff{q}$, the query algorithm selects all random vectors with an inner product larger than $\eta$ with $\bff{q}$. Then, it searches for an approximate near neighbor in the lists of points associated with these vectors in the hash table. We call \emph{buckets} the indices of the hash table (i.e. the random vectors), and \emph{filters} the function used to query the hash table (i.e. the inner product). 
In this section, we consider the asymptotic regime for $n\to\infty$, so to use the limiting distribution of the extreme concomitant in Theorem~\ref{theorem: asymptotic concominants}.

\begin{lemma}[Probability to Find a Close Point]
\label{lemma: close point probability}
For $n\to \infty$, \texttt{Top-1} \texttt{search} contains a bucket with a close point, if it exists, with at least $1-o(1)$ probability.
\end{lemma}
\begin{proof}
     Consider a close point $\bff{x}_{\alpha}$, associated to the bucket $\bff{a}_{\bff{x}_\alpha}$. That bucket is found in \texttt{search} if $\inner{\bff{a}_{\bff{x}_\alpha}}{\bff{q}} = Q_{\bff{x}_\alpha, [m]}\geq \eta$.  From Proposition~\ref{proposition: easy Gaussian tail bounds} and Theorem~\ref{theorem: asymptotic concominants}, we observe that
\begin{equation*}
    \Pr\left[Q_{\bff{x}_{\alpha}, [m]}\leq  \eta \right] = \underset{\mathcal{N}(0, 1-\alpha^2)}{\Pr}\left[Z\leq  -\sqrt{2(1-\alpha^2)\log\log m} \right]\leq \frac{1}{\log m}=O\bigg(\frac{1}{\log n}\bigg),
\end{equation*}
    where we use $m= n^{\frac{\theta}{1-\alpha^2}}$. Thus, with probability  at least $1-o(1)$, $\bff{x}_\alpha$ is associated to a vector that exceeds the threshold. 
\end{proof}

\begin{lemma}[Expected Number of Buckets and Far Points] 
\label{lemma: Expected number of far points} For  $n\to \infty$, $0\leq\beta<\alpha <1$ such that $(\alpha-\beta)=\Omega(\sqrt{\log\log n/\log n})$, 
\texttt{Top-1} \texttt{search} returns in expectation at most $n^{\theta + o(1)}$ buckets, containing in expectation at most $n^{1-\theta \frac{(\alpha-\beta)^2}{(1-\alpha^2)(1-\beta^2)}+o(1)}$ far points.
\end{lemma}

\begin{proof}
We observe that by setting $m = n^{\frac{\theta}{1-\alpha^2}}$ we get $\frac{\log \log m}{\log m} = O\big((1-\alpha^2)\frac{\log \log n}{\log n}\big)$. Thus, the threshold is $\eta  \geq \alpha\sqrt{2\log m}\big(1-O\big(\frac{1-\alpha^2}{\alpha}\sqrt{\frac{\log \log n}{\log n}}\big)\big)$, which is positive for $\alpha\geq \alpha-\beta = \Omega(\sqrt{\log\log n/\log n})$. From Proposition~\ref{proposition: easy Gaussian tail bounds}, the probability that a filter exceeds the threshold is
\begin{equation}
    \Pr[\inner{\bff{a}}{\bff{q}}\geq \eta] \leq \underset{\mathcal{N}(0,1)}{\Pr}\left[Z\geq \alpha\sqrt{2\log m}\left(1-\frac{1-\alpha^2}{\alpha}o(1)\right)\right] \leq m^{-\alpha^2 + (1-\alpha^2)o(1)},
    \label{eq: probability check vector}
\end{equation}
as the projection over a Gaussian vector is a normal random variable. In expectation, a query inspects at most $m^{1-\alpha^2+o(1-\alpha^2)}$ buckets. The claim follows by setting $m=n^{\frac{\theta}{1-\alpha^2}}$. For the analysis of far points, we may write $\eta \geq \beta\sqrt{2\log m}+ (\alpha-\beta)\sqrt{2\log m}\big(1-\frac{1-\alpha^2}{\alpha-\beta}O\big(\sqrt{\frac{\log\log n}{\log n}}\big)\big)$. The second factor is positive for $(\alpha-\beta)= \Omega(\sqrt{\log\log n/\log n})$. Thus, by applying Theorem \ref{theorem: asymptotic concominants} and Proposition \ref{proposition: easy Gaussian tail bounds}, the probability to inspect a far point ${\bf x}_{\beta}$ is
\begin{align}
\Pr[Q_{\bff{x}_{\beta}, [m]}\geq  \eta] &\leq \underset{\mathcal{N}(\beta\sqrt{2\log m}, 1-\beta^2)}{\Pr}\left[Z\geq  \beta\sqrt{2\log m}+(\alpha-\beta)\sqrt{2\log m}\left(1-\frac{1-\alpha^2}{\alpha-\beta}o(1)\right)\right]\notag \\
&\leq \exp \bigg[-\frac{(\alpha-\beta)^2}{1-\beta^2} \left(1-\frac{1-\alpha^2}{\alpha-\beta}o(1)\right)^2\cdot \log m\bigg]\notag\\
&=m^{-\frac{(\alpha-\beta)^2}{1-\beta^2}+\frac{(1-\alpha^2)(\alpha-\beta)}{1-\beta^2}o(1)}, \label{eq: far point}
\end{align}
By inserting $m=n^{\frac{\theta}{1-\alpha^2}}$ in the previous inequality, we obtain $n^{-\theta \frac{(\alpha-\beta)^2}{(1-\alpha^2)(1-\beta^2)}+o(1)}$, as $\frac{\alpha-\beta}{1-\beta^2}=O(1)$. Since we have at most $n$ far points, the expected number of inspected far points is at most $n^{1-\theta \frac{(\alpha-\beta)^2}{(1-\alpha^2)(1-\beta^2)}+o(1)}$.
\end{proof}
The proposed data structure \texttt{Top-1} (Algorithm \ref{alg: Top-1 Data Structure}) is a naive solution with high space, pre-processing time and query time, for the $(\alpha,\beta)$-ANN.
 
\begin{theorem} 
\label{theorem: Top-1}
Consider $n\rightarrow\infty$. For any $0\leq\beta<\alpha<1$ such that $(\alpha-\beta)=\Omega(\sqrt{\log\log n/\log n})$, $0<\theta\leq O(1)$,  and for any dataset $\mathcal{S}=\{x_i\}_{i=1,\dots, n}$ in $\mathbb{S}^{d}$, \texttt{Top-1} solves with at least $1-o(1)$ probability the $(\alpha,\beta)$-ANN using pre-processing time $O(d\cdot n^{1+\frac{\theta}{1-\alpha^2}})$, space $O(d\cdot \max\{n, n^{\frac{\theta}{1-\alpha^2}}\})$, and expected query time $O(d\cdot \max\{n^{\frac{\theta}{1-\alpha^2}},n^{1-\theta \frac{(\alpha-\beta)^2}{(1-\alpha^2)(1-\beta^2)}+o(1)}\})$. 
\end{theorem}

\begin{proof}
 The pre-processing time is given by $O(d\cdot n\cdot m) = O(d\cdot n^{1+\frac{\theta}{1-\alpha^2}})$ as for each of the $n$ points we need to look at $m$ random vectors of dimensionality $d$.  
 Each point is assigned to only one random vector, so the space needed to store the data structure is $O(d\cdot (n + m)) = O(d\cdot 
\max\{n, n^{\frac{\theta}{1-\alpha^2}}\})$. The running time is given by summing the running time of \texttt{search} and \texttt{query}. The buckets in \texttt{search}$(\bff{q})$ are found in time \( O(d \cdot m )\) while the expected running time of \texttt{query}$(\bff{q})$ is given by the expected number of far points present in these buckets, which are at most $n^{1-\theta \frac{(\alpha-\beta)^2}{(1-\alpha^2)(1-\beta^2)}+o(1)}$, and the expected number of buckets returned by \texttt{search}, which are at most $n^{\theta + o(1)}$, due to Lemma \ref{lemma: Expected number of far points}. Thus, the expected running time is at most $O(d\cdot \max\{n^{\frac{\theta}{1-\alpha^2}},n^{1-\theta \frac{(\alpha-\beta)^2}{(1-\alpha^2)(1-\beta^2)}+o(1)}, n^{\theta + o(1)}\})$.
The problem is solved with at least $1-o(1)$ probability due to Lemma~\ref{lemma: close point probability}.
\end{proof}

\subsection{Balanced and Unbalanced Top-1}
\label{subsection: Balanced and Unbalanced Top-1}
The standard way to minimize the expected query time of an algorithm that solves $(\alpha, \beta)$-ANN is to balance the number of buckets that have to be inspected with the number of far points (``error'') that are associated witht those buckets. To balance the contribution of far points and the number of buckets inspected we choose $\theta = \frac{(1-\alpha^2)(1-\beta^2)}{(1-\alpha\beta)^2}$, which solves the equation $\theta = 1 - \theta \frac{(\alpha - \beta)^2}{(1-\alpha^2)(1-\beta^2)}$. We denote this specific solution as $\rho$ to highlight its connection to standard ANN analysis. However, alternative values of $\theta$ can be chosen to achieve different trade-offs.
The next corollary follows from Lemma \ref{lemma: Expected number of far points}.

\begin{corollary}[Balanced and Unbalanced \texttt{Top-1}] \label{cor:balunbal}Consider $n\rightarrow\infty$. For any $0 \leq \beta<\alpha<1$ such that $(\alpha-\beta)=\Omega(\sqrt{\log\log n/\log n})$, consider $\sigma = 2\frac{(1-\alpha^2)(1-\beta^2)}{(1-\alpha\beta)^2+(\alpha-\beta)^2}$ and $\rho = \frac{(1-\alpha^2)(1-\beta^2)}{(1-\alpha\beta)^2}$. Define balanced and unbalanced \texttt{Top-1} as the data structures initialized with $\theta=\rho$ and $\theta=\sigma$ respectively. Then, 
\begin{enumerate}
    \item The expected number of buckets inspected by \texttt{search}$(\bff{q})$ is at most  $n^{\rho + o(1)}$ for balanced \texttt{Top-1}, and  $n^{\sigma + o(1)}$  for unbalanced \texttt{Top-1}.
    \item The buckets from \texttt{search}$(\bff{q})$ contain, in expectation, at most $n^{\rho + o(1)}$  far points for balanced \texttt{Top-1} and  $n^{\frac{\sigma}{2} + o(1)}$ far points for unbalanced \texttt{Top-1}.
    \item For any $0\leq\beta<\alpha\leq 1$ we have $\frac{\sigma}{2} < \rho<\sigma$.
\end{enumerate}
\end{corollary}
\begin{proof}
 It follows by a simple computation from Lemma \ref{lemma: Expected number of far points}.
\end{proof}
As $\frac{\rho}{1-\alpha^2} \geq 1$ the space and the running time of balanced \texttt{Top-1} is $O(d\cdot n^{\frac{\rho}{1-\alpha^2}})$, the latter considers the worst case as the expected number of far points is $n^{\rho + o(1)}\ll n^{\frac{\rho}{1-\alpha^2}}$. Space and running time follow directly for unbalanced \texttt{Top-1} as $\sigma > \rho$.
The unusual behavior of unbalanced \texttt{Top-1} will be further clarified in the next section, where its utility for $(\alpha,\beta)$-DP-ANNC will be discussed. 
Although \texttt{Top-1} benefits from a clean and straightforward analysis (due to the assumption  $n \to \infty$), it involves significant preprocessing, space, and query time requirements. These limitations will be addressed in Section~\ref{subsection: tensorization} through the use of \emph{tensorization}. Additionally, the assumption of  $n \to \infty$  will be lifted with a minor modification to the algorithm, as detailed in Section~\ref{subsection: analysis of CloseTop-1}.
Nevertheless, as we will show in the next section, \texttt{Top-1} still provides a meaningful solution for $(\alpha, \beta)$-ANNC under differential privacy constraints.
\section{From ANN to DP-ANNC} 
\label{section: From ANN to DP-ANNC}
In this section, we study the relationship between ANN and DP-ANNC. 
We expose a general way to solve ANNC starting from a \emph{space-partitioning data structure} for ANN, and discuss different differentially private mechanisms to privatize the ANNC data structure. We also show how the unbalanced data structure from the previous section can be used to increase the accuracy for $(\alpha, \beta)$-DP-ANNC, by paying an increase in query time, pre-processing time, and space usage.

\subsection{From ANN to ANNC}

\begin{algorithm}[t]
\footnotesize
\caption{From ANN to ANNC using a space-partitioning data structure} \label{alg: reduction}
\begin{minipage}[t]{0.52\textwidth}
    \begin{algorithmic}[1]
    \Procedure{\texttt{construction}$(\mathcal{L})$}{} 
    \State $T[1,\dots, m]\gets (0,\dots, 0)$
    \For{$\textbf{each } L_i \in \mathcal{L}$}
        \State $T[i]\gets |L_i|$  
    \EndFor
    \State \textbf{return } $\mathcal{D}=(T, \mathcal{L})$
\EndProcedure
    \end{algorithmic}
    \end{minipage}
    \hfill
    \begin{minipage}[t]{0.46\textwidth}
    \begin{algorithmic}[1]
       \Procedure{$\texttt{query}(\bff{q}, \mathcal{Q})$}{} 
    \State $I(\bff{q}) \gets \mathcal{Q}(\bff{q})$ 
    \State $\widehat{\text{ans}}\gets 0$
    \For {$i \in I(\bff{q})$}
        \State $\widehat{\text{ans}}\gets \widehat{\text{ans}}+T[i]$
    \EndFor
    \State \textbf{return }$\widehat{\text{ans}}$
\EndProcedure
        \end{algorithmic}
    \end{minipage}
\end{algorithm}
Inspired by the \emph{list-of-points} data structure developed in~\cite{andoni2017optimal}, we define a family of data structures suitable for a general reduction from ANN to ANNC. 
\begin{definition}[Space Partitioning Data Structure]
Given a set $\mathcal{S} \subseteq \mathbb{S}^{d-1}$ and an integer $m$, a \emph{space-partitioning data structure} for the ANN problem is defined as follows:
\begin{itemize}
\item The data structure is a partition\footnote{Technically we define the data structure by using a not full partition, as we allow some points to not be stored.} of $\mathcal{S}$ into $m$ sets $\mathcal{L} = (L_1, \ldots, L_m)$ such that
$\bigcup_{i \in [m]} L_i \subseteq \mathcal{S}$ and $L_i \cap L_j = \emptyset$ for $i \neq j$, and a function $\mathcal{Q}$ that maps $\mathbb{S}^{d-1}\ni \bff{q}   \mapsto I(\bff{q}) \subseteq [m]$.
\item For a query $\bff{q}$, we obtain the set $I(\bff{q}) \gets\mathcal{Q}(\bff{q})$ and scan all points in $L_i, i \in I(\bff{q})$. If there exists a point with inner product at least $\beta$, we return it. Otherwise we return $\perp$.
\end{itemize}
The total space is $ |\mathcal{Q}| + O(d\cdot n)$, where $|\mathcal{Q}|$ is the space necessary to store the function $\mathcal{Q}$. The query time is at most $T_\mathcal{Q}(\bff{q}) + O\big(d\sum_{i \in I(\bff{q})} |L_i|\big)$, where $T_\mathcal{Q}$ is the time taken to compute $I(\bff{q})$ given query $\bff{q}\in \mathbb{S}^{d-1}$, and $O\big(d\sum_{i \in I(\bff{q})} |L_i|\big)$ is the worst-case time needed to check all the points.
\end{definition}
For example, for Algorithm~\ref{alg: Top-1 Data Structure}, $L_i$ represents the points that achieve their maximum inner product with $\bff{a}_i$. $\mathcal{Q}$ consists of all $\bff{a}_1, \ldots, \bff{a}_m$ (its size is $dm$) and computes the indices $I(\bff{q})$ of all filters that are above the query threshold. For the data structure of Andoni et al.~\cite{andoni2024differentially} discussed in the introduction, $L_i$ are the leaves of the tree, and $\mathcal{Q}$ represents the navigation tree-based data structure.
Algorithm~\ref{alg: reduction} presents a simple transformation for a space-partitioning data structure:
It indeed suffices to substitute the actual points with their amount in each list.
The new $(\alpha,\beta)$-ANNC \texttt{query} returns the sum of the elements contained in these lists. Since each point is stored at most once, summing the cardinality of each bucket ensures that no point is counted more than once. 

\begin{lemma}[From ANN to ANNC]
\label{lemma: Reduction from ANN to ANNC}
Let $\mathcal{S}=\{x_i\}_{i=1,\dots, n} \subseteq \mathbb{S}^{d-1}$. Consider a space-partitioning data structure for $\mathcal{S}$ such that for each $\bff{q} \in \mathbb{S}^{d-1}$: (i) $I(\bff{q})$ contains a list with a close point with probability at least $1-o(1)$, and (ii) the expected number of far points in $\bigcup_{i \in I(\bff{q})} L_i$ is at most $\mathcal{K}$.
Then \texttt{query} in Algorithm~\ref{alg: reduction} returns a value $\widehat{\text{ans}}$ that,  with probability at least $2/3$, satisfies the following inequality:
\begin{equation*}
    \big(1-o(1)\big)|\mathcal{S}\cap B(\bff{q}, \alpha)|\leq \widehat{\textnormal{ans}}\leq |\mathcal{S}\cap B(\bff{q}, \beta)|+\mathcal{K}. 
\end{equation*}
The data structure in Algorithm~\ref{alg: reduction} uses space $|\mathcal{Q}| + O(n)$ and the query time is $T_\mathcal{Q}(\bff{q}) + |I(\bff{q})|$.
\end{lemma}

\subsection{From ANNC to DP-ANNC} 
The data structure returned by Algorithm~\ref{alg: reduction} uses  \emph{counters}, which is essentially a histogram. This histogram can be privatized using the algorithms presented in Section~\ref{subsection: differential privacy}.
To achieve differential privacy, we need to analyze the sensitivity of the counters $T[1..m]$.
Let $\mathcal{S}$ and $\mathcal{S}'$ be two neighboring datasets that differ in exactly one point, and let $\mathcal{D}$ and $\mathcal{D}'$ be the data structures constructed for these data sets, respectively. 
Applying Algorithm~\ref{alg: reduction} to $\mathcal{D}$ and $\mathcal{D}'$ will result in two counters $T$ and $T'$ that differ by at most 1 in at most one position.
If $\mathcal{Q}$ is data independent, i.e., $\mathcal{Q}$ does not depend on the actual data set $\mathcal{S}$, it is sufficient to privatize the counter $T$, which can be done using any differentially private mechanism \texttt{make\_private} for histograms, 
as shown in the aforementioned \texttt{DPTop-1} (see Algorithm \ref{alg: DPTop-1 Data Structure}). 
The next theorem states the guarantees for two specific privacy mechanisms.

\begin{theorem}[DP-ANNC with Truncated Laplace or Max Projection] 
\label{theorem: DP-ANNC}
Let $\mathcal{S}=\{x_i\}_{i=1,\dots, n} \subseteq \mathbb{S}^{d-1}$, $\bff{q} \in \mathbb{S}^{d-1}$. Consider a space partitioning data structure for $\mathcal{S}$ satisfying the assumptions of Lemma \ref{lemma: Reduction from ANN to ANNC}, with the addition of $\mathbb{E}[|I(\bff{q})|]\leq \mathcal{K}$ and $\mathcal{Q}$ being data independent. 
When $T$ is privatized using the truncated Laplace mechanism, the data structure is $(\varepsilon, \delta)$-DP, for any $\varepsilon \leq 1$, requires an additional $O(n)$ term in space and pre-processing time (compared to Lemma~\ref{lemma: Reduction from ANN to ANNC}), and the query algorithm returns a value $\widetilde{\textnormal{ans}}$ that satisfies the following inequality with probability at least $2/3$:
\begin{equation}
\label{eq: DP result}
    (1-o(1))|\mathcal{S}\cap B(\bff{q}, \alpha)|-O\bigg(\frac{\log(1/\delta)}{\varepsilon}\mathcal{K}\bigg)\leq \widetilde{\textnormal{ans}}\leq |\mathcal{S}\cap B(\bff{q}, \beta)|+O\bigg(\frac{\log(1/\delta)}{\varepsilon}\mathcal{K}\bigg).
\end{equation}
When $T$ is privatized with the \texttt{Max Projection} mechanism, then the data structure is $(\varepsilon, 0)$-DP, requires an additional $O(\varepsilon\cdot n)$ term in space and pre-processing time, and the additive error in  \autoref{eq: DP result}  becomes $O(\frac{\mathcal{K}}{\varepsilon})$.
\end{theorem}
Due to Lemma \ref{lemma: close point probability} and Corollary \ref{cor:balunbal}, balanced \texttt{Top-1} satisfies the requirements for Theorem~\ref{theorem: DP-ANNC} with $\mathcal{K} = n^{\rho +o(1)}$, which proves Theorem~\ref{theorem: our contribution DP-ANNC}.
Moreover, the tree-based data structure described by Andoni et al.~\cite{andoni2024differentially} is another data structure that satisfies these requirements of Theorem~\ref{theorem: DP-ANNC}.\footnote{Technically, as stated earlier, \cite{andoni2024differentially} analyze their data structure for a fixed choice of $\alpha$.}

\subsubsection{Usefulness of Unbalanced Data Structure} 
We now study how an unbiased and uncorrelated differentially private estimator of $T$, from the unbalanced \texttt{Top-1} ANNC data structure, can be used to reduce the error compared to Theorem~\ref{theorem: DP-ANNC}. The construction leverages the concentration of the sum of i.i.d. Laplace random variables.

\begin{theorem}[DP-ANNC with Laplace Noise and Unbalanced Data Structure] 
\label{theorem: DP-ANNC Unbalanced}
Let $\mathcal{S}=\{x_i\}_{i=1,\dots, n} \subseteq \mathbb{S}^{d-1}$, $\bff{q} \in \mathbb{S}^{d-1}$. Consider a space partitioning data structure for $\mathcal{S}$ satisfying the assumptions of Lemma \ref{lemma: Reduction from ANN to ANNC}, with the addition of $\mathbb{E}[|I(\bff{q})|]\leq \mathcal{K}^2$ and $\mathcal{Q}$ being data independent.
When $T$ is privatized by using the Laplace mechanism, the data structure is $(\varepsilon, 0)$-DP, for any $\varepsilon\leq 1$, and \texttt{query} returns a value $\widetilde{\textnormal{ans}}$ that satisfies the following inequality with probability at least $2/3$:
\begin{equation*}
    (1-o(1))|\mathcal{S}\cap B(\bff{q}, \alpha)|-O\bigg(\frac{\mathcal{K}}{\varepsilon}\bigg)\leq \widetilde{\textnormal{ans}}\leq |\mathcal{S}\cap B(\bff{q}, \beta)|+O\bigg(\frac{\mathcal{K}}{\varepsilon}\bigg).
\end{equation*} 
 The privatized data structure requires an additional $O(m)$ space, and pre-processing time.
\end{theorem}
Due to Lemma \ref{lemma: close point probability} and Corollary \ref{cor:balunbal}, unbalanced \texttt{Top-1} satisfies the requirements for Theorem \ref{theorem: DP-ANNC Unbalanced} with $\mathcal{K} = n^{\frac{\sigma}{2} +o(1)}$. 
As $\frac{\sigma}{2}<\rho$ unbalanced \texttt{Top-1} is always more accurate than \texttt{Top-1} for DP-ANNC.
However, the pre-processing time and the space increase. 

In the next section we provide several improvements for \texttt{Top-1}, aiming to get rid of the asymptotic assumption $n\rightarrow \infty$ used to apply the Theorem \ref{theorem: asymptotic concominants} for concomitant statistics, and reduce the pre-processing time, the query time, and the space. These improvements regard only the balanced ANN data structure; the additional requirements in space and pre-processing time of $O(m)$ for DP-ANNC with unbalanced data structures will still be present and are an interesting open question for future work. 
Finally, we highlight that for errors of the form $n^{C+o(1)}$ we may increase the range of the privacy budget to $\varepsilon\leq n^{o(1)}$ in Theorems \ref{theorem: DP-ANNC} and \ref{theorem: DP-ANNC Unbalanced}, for the same argument provided by Andoni et al. \cite{andoni2024differentially}.

\section{Improving the \texttt{Top-1} Data Structure}
\label{section: Analysis}
In this section we propose two improvements of \texttt{Top-1}. With \texttt{CloseTop-1} we get rid of the assumption of $n\to \infty$,\footnote{We observe that, although the time and space complexities are expressed in big-O notation (i.e., with a notation asymptotic in $n$), the correctness of this algorithm does not require assuming a limiting distribution for the concomitants, while this was the case for \texttt{Top-1}.} while with \texttt{TensorCloseTop-1} we reduce the pre-processing time to $d\cdot n^{1+o(1)}$, the space to $O(d\cdot n)$, and obtain an expected query time of $d\cdot n^{\theta + o(1)}$. In addition, we discuss how \texttt{TensorCloseTop-1} can solve the $(r,c)$-ANN in the Euclidean space in Appendix \ref{appendix: Embedding into the Euclidean sphere}. These improvements do not alter the core of the data structure (a hash table of points with a \texttt{search} function), allowing them to be utilized for DP-ANNC as discussed in the previous section.

\subsection{\texttt{CloseTop-1}}
\label{subsection: analysis of CloseTop-1}
\begin{algorithm}[t]
\footnotesize
\caption{\texttt{CloseTop-1 Data Structure}}\label{alg: CloseTop-1 Data Structure}
\begin{algorithmic}[1]
\Procedure {\texttt{Construction}$(\mathcal{S}\subset \mathbb{S}^{d-1}, \alpha, \beta, \theta)$}{}
\State $m\gets \big\lceil n^{\frac{\theta}{1-\alpha^2}}\big\rceil$
\State $\mathcal{A}^{m} = (\bff{a}_1, \dots, \bff{a}_m)$ with $\bff{a}_i \sim \mathcal{N}(0,1)^d$ \Comment{Sample $m$ Gaussian Random Vectors}
\State $H \gets \text{empty hash table}$
\For {$\bff{x} \in \mathcal{S}$}
    \For {$i \in \{1,\dots, m\}$}
        \If {$\sqrt{2\log m}-\frac{3}{2}\frac{\log\log m}{\sqrt{2\log m}}\leq \inner{\bff{a}_i}{\bff{x}}\leq \sqrt{2\log m}$}
            \State $H.\texttt{insert}(i, \bff{x})$
            \State \textbf{break}
        \EndIf
    \EndFor
\EndFor
 \State $\eta \gets \alpha\sqrt{2\log m}-\sqrt{2(1-\alpha^2)\log\log m}$
\State \textbf{return }$\mathcal{D} = (H, \mathcal{A}^m, m, \eta)$
\EndProcedure
\newline
\State $\triangleright$ \texttt{search} and \texttt{query} are the same in Algorithm \ref{alg: Top-1 Data Structure}.
\end{algorithmic}
\end{algorithm}
We now study \texttt{CloseTop-1} (see Algorithm \ref{alg: CloseTop-1 Data Structure}), a practical implementation of the previous asymptotic data structure. In \texttt{Top-1} we associate to each point of the dataset the random vector with the highest inner product. This is an intuitive choice that leads to a simple and clear analysis by analyzing Gaussian tails of concomitant statistics (Theorem \ref{theorem: asymptotic concominants}). However, this is an asymptotic theorem, results from the fact that $X_{\bff{x}}=\max_{\bff{a}\in \mathcal{A}^m}\inner{\bff{a}}{\bff{x}} = \sqrt{2\log m}-o(1)$ and $\text{Var}[X_{\bff{x}}] = o(1)$ \cite{hall1979rate}. In fact, it can be obtained by Lemma \ref{lemma: relation between concominant and extreme} by setting $X_{\bff{x}} = \sqrt{2\log m}$.
The intuition of \texttt{CloseTop-1} is to provide a lower and an upper bound for $X_{\bff x}$ by construction, 
by associating to each point of the dataset a random vector with an inner product \emph{close} to the expected maximum, so the name of the data structure.  
In the proposed construction, we sample $m$ random vectors $\bff{a}\sim\mathcal{N}(0,1)^d$, and we associate to any $\bff{x}$ the first random vector such that $\sqrt{2\log m}-\frac{3}{2}\frac{\log\log m}{\sqrt{2\log m}}\leq\inner{\bff{a}}{\bff{x}}\leq \sqrt{2\log m}$. 
If at least one random vector succeeds in the association, then we say that $\bff{x}$ collided. The key property of \texttt{CloseTop-1} is that a point collides with high probability (Lemma \ref{lemma: collision}), which allows to state the following lemma.
\begin{lemma} 
\label{lemma: CloseTop-1 lemma 1}
Lemma \ref{lemma: close point probability} and Lemma \ref{lemma: Expected number of far points} remain valid for \texttt{CloseTop-1} under the same assumptions, without relying on the  $n \to \infty$  condition.
\end{lemma}
As Theorem \ref{theorem: Top-1} and Corollary \ref{cor:balunbal} are derived from Lemmas \ref{lemma: close point probability} and \ref{lemma: Expected number of far points}, \texttt{CloseTop-1} give the same results as \texttt{Top-1}, without relying on the assumption that the concomitants follow a limiting distribution.

\subsection{\texttt{TensorCloseTop-1}}
\label{subsection: tensorization}
In this section, we propose \texttt{TensorCloseTop-1} (see Algorithm \ref{alg: TensorCloseTop-1 Data Structure}), to reduce the pre-processing time to $d\cdot n^{1+ o(1)}$, space to $O(d\cdot n)$, and expected query time to $d\cdot n^{\rho + o(1)}$ (for the balanced data structure). The data structure uses a technique developed in \cite{christiani2017framework} called \emph{tensoring} that essentially allows to simulate an exponential number of vectors by concatenating a polynomial number of data structures. The same expedient was used in \cite{aumuller2022sampling} to get a pre-processing time of $n^{1+\rho + o(1)}$. This technique is similar to creating a tree, yet 
this data structure allows parallel evaluation for the hashes (Line 2-3 Algorithm \ref{alg: TensorCloseTop-1 Data Structure} \texttt{search}).
Define the concatenation factor $t\in \mathbb{N}$ and assume $m^{1/t}$ is an integer; consider $t$ independent \texttt{CloseTop-1} data structures $\mathcal{D}_1, \dots, \mathcal{D}_t$ each using $m^{1/t}$ Gaussian vectors $\bff{a}_{i, j}$, where $i\in [t]$ indicates the data structure and $j \in [m^{1/t}]$ indicates the vector. For each point $\bff{x} \in \mathcal{S}$ consider the $t$ colliding vectors in each data structure $(\bff{a}_{1,i_1}, \dots, \bff{a}_{t, i_t})$, then map the point to a bucket $(i_1, \dots, i_t)\in [m^{1/t}]^t$ using a hash table. Given a query $\bff{q}\in \mathbb{S}^{d-1}$, for each data structure the indices $\tilde{B}_i$ of the random vectors are selected such that $\inner{\bff{a}}{\bff{q}}\geq \eta$, hence $\tilde{B}_i := \{j \in [m^{1/t}]: \inner{\bff{a}_{i, j}}{\bff{q}}\geq \eta\}$, and a search in all the buckets $\tilde{B}_{1}\times \dots \times \tilde{B}_t$ is performed. The number of random vectors to sample is $t\cdot m^{1/t}$ which is sub-linear in $n$ provided the data structure is not designed to search for points with exceedingly high inner product similarity. 

\begin{algorithm}[t]
\footnotesize
\caption{\texttt{TensorCloseTop-1 Data Structure}}\label{alg: TensorCloseTop-1 Data Structure}
\begin{minipage}[t]{0.51\textwidth}
    \begin{algorithmic}[1]
\Procedure{\texttt{construction}$(\mathcal{S}\subset \mathbb{S}^{d-1}, \alpha, \beta, \theta)$}{}
    \State $t\gets \big\lceil\frac{\log^{1/8}n}{1-\alpha^2}\rceil$
    \State $\tilde{m}\gets \lceil n^{\frac{1}{t}\frac{\theta}{1-\alpha^2}}\rceil$
    \State $\{\mathcal{D}_i\}_{i=1,\dots, t}\gets \{\texttt{CloseTop-1}(\mathcal{S}, \tilde{m})\}_{i=1,\dots, t}$ 
    \State $H \gets \text{empty hash table}$
    \For {$\bff{x} \in \mathcal{S}$}
        \State $i \gets (H_1\texttt{.hash}(\bff{x}),\dots, H_t\texttt{.hash}(\bff{x}))$ 
        \State $H\texttt{.insert}(i, \bff{x})$
    \EndFor
    \State $\eta \gets \alpha\sqrt{2\log \tilde{m}}-\sqrt{2(1-\alpha^2)\log\log \tilde{m}}$
    \State \textbf{return } $(H, \{\mathcal{D}\}_{i=1,\dots, t}, \eta)$
\EndProcedure
    \end{algorithmic}
    \end{minipage}
    \hfill
    \begin{minipage}[t]{0.49\textwidth}
    \begin{algorithmic}[1]
    \Procedure{\texttt{search}$(\bff{q})$}{}
     \For {$i \in (1,\dots, t)$}
        \State $\Tilde{B}_i \gets \mathcal{D}_i\texttt{.search}(\bff{q}, \eta)$
    \EndFor
    \State $\textbf{return } \Tilde{B}_1 \times \dots\times \Tilde{B}_t$ \Comment{Cartesian Product}
    \EndProcedure
    \newline
      \Procedure{$\texttt{query}(\bff{q})$}{} 
      \State $B \gets H.\texttt{search}(\bff{q})$
    \For {$i \in B$}
        \For {$\bff{x}\in H[i]$}
            \If {$\inner{\bff{q}}{\bff{x}}\geq \beta$}
                \State \textbf{return } $\bff{x}$
            \EndIf
        \EndFor
    \EndFor
    \State \textbf{return None} 

\EndProcedure
        \end{algorithmic}
    \end{minipage}
\end{algorithm}
\begin{proposition}[Tensorization] 
\label{proposition: tensorization}
For any constant $C>0$ assume $\frac{1}{1-\alpha^2}\leq (\log n)^{C}$. Then for $t= \frac{\log^{1/8} n}{1-\alpha^2}$, $m = n^{\frac{\theta}{(1-\alpha^2)}}$, and $\theta =O(1)$, we have $t\cdot m^{1/t} = n^{o(1)}$
\begin{proof}
    Just a simple computation: $m^{1/t} = n^{\frac{\rho}{\log^{1/8}n}} = n^{o(1)}$ as $\theta = O(1)$, while $t=\frac{\log^{1/8}n}{1-\alpha^2}\leq (\log n)^{O(1)} = n^{o(1)}$, then $t\cdot m^{1/t} = n^{o(1)}$.
\end{proof}
\end{proposition}
In practice $t, m^{1/t}$ and $m$ are all integers. However, this does not affect the asymptotic behavior since $\lceil t \rceil m^{1/\lceil t \rceil}\leq(t+1)\cdot m^{1/t}=n^{o(1)}$. With this trick, we reduce the query time due to \texttt{search} (Line 2,3,4 Algorithm \ref{alg: TensorCloseTop-1 Data Structure}, procedure \texttt{search}) to $n^{o(1)}$. Therefore, the query time is mainly affected by how many times the data structure needs to access the hash table, which can be bounded in expectation.
Under similar assumptions for $\alpha$ we can prove that \texttt{TensorCloseTop-1} finds with high probability a close point.
\begin{lemma}[Probability to Find a Close Point]
\label{lemma: close point tensorclosetop1}
For any $0\leq \beta <\alpha<1$ such that $1-\alpha^2 = \omega(\log^{-3/4}n)$, \texttt{TensorCloseTop-1} finds a close point, if it exists, with at least $1-o(1)$ probability.
\end{lemma}

\begin{lemma}[Expected Number of Buckets and Far Points]
\label{lemma: far points tensorclos}
For any $0\leq \beta <\alpha<1$ such that $(1-\alpha^2) = \omega(\log^{-3/4}n)$, and $(\alpha-\beta) = \Omega\big( {\scriptstyle \sqrt{\frac{\log\log n}{\log^{7/8} n}}} \big)$, \texttt{TensorCloseTop-1} \texttt{search} finds in expectation at most $n^{\theta + o(1)}$ buckets containing at most $n^{1-\theta \frac{(1-\alpha^2)(1-\beta^2)}{(\alpha-\beta)^2}+o(1)}$ far points.
\end{lemma}
 The previous Lemma states that $\theta$ has the same function it has in \texttt{CloseTop-1}, so it can be used to construct balanced and unbalanced \texttt{TensorCloseTop-1}. We now argue for the query, space and pre-processing time.
 \begin{theorem}
 \label{theorem: TensorCloseTop-1}
     For any $0\leq\beta<\alpha<1$ such that $(1-\alpha^2)=\omega(\log^{-3/4}n)$, $(\alpha-\beta)=\Omega\big({\sqrt{\frac{\log\log n}{\log^{7/8} n}}}\big)$, and $0<\theta \leq O(1)$. For any dataset $\mathcal{S}=\{x_i\}_{i=1,\dots, n}$ in $\mathbb{S}^{d}$, \texttt{TensorCloseTop-1} solves with at least $1-o(1)$ probability the $(\alpha,\beta)$-ANN using space $O(d\cdot n)$, preprocessing time $d\cdot n^{1+o(1)}$, and expected query time $d\cdot \max\{n^{\theta + o(1)}, n^{1-\theta\frac{(\alpha-\beta)^2}{(1-\alpha^2)(1-\beta)^2}+o(1)}\}$. 
 \end{theorem}
 \begin{proof}
     Due to Proposition \ref{proposition: tensorization}, the data structure needs to store $t\cdot m^{1/t} = n^{o(1)}$ random vectors. As each point $\bff{x}$ is stored in at most one bucket, the space is $O(d\cdot (n+n^{o(1)})) = O(d\cdot n)$. As to each point it is necessary to compute $m^{1/t}$ inner products at most $t$ times, the pre-processing time is $O(t\cdot d\cdot n\cdot m^{1/t}) = O(d\cdot n^{1+o(1)})$. The buckets in \texttt{search}$(\bff{q})$ can be computed in time $O(d\cdot n^{o(1)})$, so the expected query time is at most $d\cdot \max\{n^{\theta + o(1)}, n^{1-\theta\frac{(\alpha-\beta)^2}{(1-\alpha^2)(1-\beta)^2}+o(1)}\}$ due to Lemma \ref{lemma: far points tensorclos}. The problem is solved with at least $1-o(1)$ probability due to Lemma \ref{lemma: close point tensorclosetop1}.
 \end{proof}
As Corollary \ref{cor:balunbal} applies to \texttt{TensorCloseTop-1} due to Lemma \ref{lemma: far points tensorclos}, the parameters $\rho$ and $\sigma$ respectively characterize the balanced and unbalanced versions of \texttt{TensorCloseTop-1}. The balanced version has an expected query time of  $O(d \cdot n^{\rho + o(1)})$ and, when combined with Theorem \ref{theorem: DP-ANNC}, achieves an additive error of  $O(n^{\rho + o(1)}/\varepsilon)$  for differentially private approximate nearest neighbor search (DP-ANNC). For $\varepsilon = O(1)$, the additional space and preprocessing requirements are  $O(n)$.
In contrast, the unbalanced version of \texttt{TensorCloseTop-1} has an expected query time of  $O(d \cdot n^{\sigma + o(1)})$  and, when used with Theorem \ref{theorem: DP-ANNC Unbalanced}, yields an asymptotically smaller additive error  $O(n^{\frac{\sigma}{2} + o(1)}/\varepsilon)$  for DP-ANNC. However, this approach incurs a significant drawback: the Laplace noise introduces an additional space and preprocessing overhead of  $O(m) = O(\tilde{m}^t) = O(n^{\frac{\sigma}{1-\alpha^2}})$ . This overhead becomes the dominant cost as $\frac{\sigma}{1-\alpha^2} \geq 1$, especially in the Euclidean ANN problem, where $(1-\alpha^2)^{-1} = \text{polylog}(n)$. For further details on how this data structure is applied in Euclidean space, see Appendix \ref{appendix: Embedding into the Euclidean sphere}.

\section{Conclusion and Open Problems}
This paper introduced and analyzed simple linear space data structures that solve the $(\alpha, \beta)$-ANN problem and can be transformed into efficient solutions for its counting variant under differential privacy.
This provides an alternative data structure to the one proposed recently by Andoni et al.~\cite{andoni2024differentially} with a simpler data structure and analysis. 
We provided general black-box transformations from approximate near neighbor problems to their counting variant under privacy constraints and showed that interesting error/time trade-offs are possible via \emph{unbalanced} ANN data structures.
The most intriguing open question was already posed by Andoni et al.~\cite{andoni2024differentially}: Can one obtain better accuracy guarantees for range counting than by transforming near neighbor data structures that have well-understood lower bounds~\cite{andoni2017optimal}? For example, \cite{DBLP:conf/soda/AhleAP17} describes a sampling based range counting algorithm that could be a good starting point for further investigation. 
For the presented data structures, one should further investigate the relation of the noise error due to differential privacy and the error due to including ``far points'' which could give interesting trade-offs. We initiated such a study through unbalanced ANN data structures; the main obstacle for a space-efficient solution is to store ``small counts'' in a data structure that uses space $O(f(\varepsilon) n^{1 + o(1)})$ and provides unbiased counters such that the expected error of the sum of $\mathcal{K}$ counters is only a factor $O(\sqrt{\mathcal{K}})$ larger than the expected per-point error.
Finally, while we believe that our algorithms are simple and straightforward, an experimental comparison between the different solutions presented here and in the literature seems necessary, not only for approximate range counting, but also filtering-based approximate near neighbor search.
In fact, only the work of Pham et al.~\cite{pham2022falconn++} provided evidence of the practical impact of filtering-based near neighbor search, and they achieve their result by a combination of LSH and LSF.



\bibliography{bibliography}

\newpage
\appendix

\appendix
\section{Useful inequalities and Additional Definitions}
\subsection{Differentially Private Mechanisms}
\label{appendix: dp}
In this work we considered three differentially private mechanisms:
\begin{itemize}
    \item The \emph{Truncated Laplace mechanism}~\cite{geng2020tight}, used also by Andoni et al. \cite{andoni2024differentially}, which obfuscate each \emph{positive} entry of $f(\mathcal{S})$ by adding truncated Laplace noise. The mechanism is $(\varepsilon,\delta)$-DP and produces a biased estimator with expected absolute error $O(\log(1/\delta)/\varepsilon)$. In our context, it has the advantage that it only needs to sample and store the counts of the non-zero entries, i.e., at most $n$ random variables.
    \item  The \emph{Laplace mechanism}~\cite{dwork2006calibrating}, which adds independent Laplace noise to each entry of $f(\mathcal{S})$. The mechanism is $(\varepsilon, 0)$-DP and produces an unbiased estimator with uncorrelated entries with absolute expected error $O(1/\varepsilon)$. The estimator behaves well for \emph{range queries} (i.e. $\sum_{i\in B}f_{i}(\mathcal{S})$ for some $B \subseteq [|\mathcal{X}|]$) obtaining an expected absolute error $O(\sqrt{B}/\varepsilon)$. However, it requires to sample and store $|\mathcal{X}|$ random variables.
    \item The \emph{Max Projection mechanism}~\cite{aumuller2022representing} which stores all ``small counts'' in a sketching-based data structure. The mechanism is $(\varepsilon, 0)$-DP and produces a data structure with $O(1)$ access time returning a biased estimator with expected absolute error $O(1/\varepsilon)$. The additional space it needs is $O(\varepsilon \cdot n)$, making it a valid pure-DP alternative to the Truncated Laplace mechanism.
\end{itemize}
\subsection{Tail Bounds}
\label{subsection: Tail bounds}
 We will make use of the following Gaussian tail bounds.
\begin{proposition}[Gaussian Tail Bounds \cite{dubhashi2009concentration}]
\label{proposition: easy Gaussian tail bounds}
    Let $Z\sim \mathcal{N}(\mu, \sigma^2)$. Then, for any $t\geq 0$, we have that $
    \Pr[|Z-\mu|\geq t]\leq e^{-\frac{t^2}{2\sigma^2}}$.
\end{proposition}
\begin{proposition}[Proposition 3, \cite{szarek1999nonsymmetric}]
\label{proposition: Gaussian tail bounds}
Let $Z$ be a standard normal random variable. Then, for any $t> -1$, we have that
\begin{equation*}
    \frac{2\sqrt{2\pi}}{t+\sqrt{t^2+4}}e^{-\frac{t^2}{2}}\leq \Pr[Z\geq t]\leq \frac{4\sqrt{2\pi}}{3t+\sqrt{t^2+8}}e^{-\frac{t^2}{2}},
\end{equation*}
\end{proposition}
From the previous proposition, it may be more useful to use the following loose bounds:
\begin{proposition}
\label{proposition: loose tail gaussian}
    Let $Z$ be a standard normal random variable. Then, for any $t>1$, we have that $\frac{2\sqrt{2\pi}}{(1+\sqrt{5})t}e^{-\frac{t^2}{2}}\leq\Pr[Z\geq t]\leq \frac{4\sqrt{2\pi}}{3t}e^{-\frac{t^2}{2}}$
\end{proposition}
\begin{proof}
    The bounds follow from Proposition \ref{proposition: Gaussian tail bounds}. The upper bound is trivial, while the lower bound follows by noticing that $\sqrt{t^2+4}\leq t\sqrt{5}$ as $t>1$.
\end{proof}

\section{Omitted Proofs}\label{sec:omittedproof}
\subsection{Omitted Proofs in Section \ref{section: From ANN to DP-ANNC}}
\begin{proof}[Proof of Lemma \ref{lemma: Reduction from ANN to ANNC}] 
    We start with the lower bound. Let $X$ be the random variable indicating the number of close points in $\mathcal{S}$ not included in $\widehat{\text{ans}}$. Due to requirement (i) the probability to not find, and so to not count, a close point is at most $o(1)$, then $\mathbb{E}[X] \leq |\mathcal{S}\cap B(\bff{q}, \alpha)|o(1)$. Using Markov's inequality we have that $X\leq |\mathcal{S}\cap B(\bff{q}, \alpha)|o(1)$ with constant probability. Consider now the number of close points counted $\widehat{\text{ans}}_{\text{close}}$, clearly $\widehat{\text{ans}}\geq \widehat{\text{ans}}_{\text{close}}$ and $\widehat{\text{ans}}_{\text{close}} = |\mathcal{S}\cap B(\bff{q}, \alpha)|-X$. Therefore, with constant probability we have $\widehat{\text{ans}}_{\text{close}}\geq |\mathcal{S}\cap B(\bff{q}, \alpha)|-|\mathcal{S}\cap B(\bff{q}, \alpha)|o(1) = |\mathcal{S}\cap B(\bff{q}, \alpha)|(1-o(1))$ which concludes the proof for the lower bound.
    
    We proceed  with the upper bound. Let $Y$ be the random variable indicating the number of far points in $\mathcal{S}$ included in $\widehat{\text{ans}}$, then $\widehat{\text{ans}} \leq |\mathcal{S}\cap B(\bff{q},\beta)|+Y$. Due to requirement (ii) we have that $\mathbb{E}[Y]\leq \mathcal{K}$.
    Thus, by using Markov's inequality $\widehat{\text{ans}} \leq |\mathcal{S}\cap B(\bff{q},\beta)|+ \mathcal{K}$ with constant probability. Combining these two bounds, we arrive at the desired result.

    As the algorithm substitute $d$ dimensional point with a number, the space to store these number reduces to $O(n)$. The query does not search for a ANN, but sums all the numbers stored in the counter on the indices $I(\bff{q})$, so the running time is $T_{\mathcal{Q}}(\bff{q})+|I(\bff{q})|$.
\end{proof}
\begin{proof}[Proof of Theorem \ref{theorem: DP-ANNC}] As $\mathcal{Q}$ is data independent, on neighboring datasets the data structures differ only in the counters. We start by considering the truncated Laplace noise. Let $T$ be the counter from Algorithm \ref{alg: reduction} and $\tilde{T}$ be the differentially private version. The error due to differential privacy in the counts is $|\Tilde{T}[i]-T[i]|\leq O\big(\frac{\log(1/\delta)}{\varepsilon}\big)$, as the truncated Laplace mechanism adds bounded noise sampled from $[-C\frac{\log(1/\delta)}{\varepsilon}, C\frac{\log(1/\delta)}{\varepsilon}]$ for some $C>0$. Therefore, the expected error between $\widetilde{\text{ans}}$ and $\widehat{\text{ans}}$ is at most
    \begin{equation*}
        \mathbb{E}[|\widetilde{\text{ans}}-\widehat{\text{ans}}|] = \mathbb{E}\left[\left|\sum_{i\in I(\bff{q})}(\tilde{T}[i]-T[i])\right|\right]\leq O\bigg(\frac{\log(1/\delta)}{\varepsilon}\mathbb{E}[|I({\bf q})|]\bigg) \leq O\left(\frac{\log(1/\delta)}{\varepsilon}\mathcal{K}\right).
    \end{equation*}
     Thus, by Markov's inequality we have $|\widetilde{\text{ans}}-\widehat{\text{ans}}|\leq O\big(\frac{\log(1/\delta)}{\varepsilon}\mathcal{K}\big)$ with constant probability. The claim follows by Lemma \ref{lemma: Reduction from ANN to ANNC} and $\varepsilon\leq 1$. As the Truncated Laplace mechanism only needs to sample at most $n$ random variables, the additional factor in space and pre-processing time is $O(n)$.

    \texttt{Max Projection} returns a $(\varepsilon, 0)$-DP counter $\tilde{T}$ with constant access time, using space and pre-processing time  $O(\varepsilon n)$, and with error $\mathbb{E}[|T[i]-\tilde{T}[i]|]\leq O(1/\varepsilon)$ (Corollary 8.3 \cite{aumuller2022representing}). The analysis then follows identically. 
\end{proof}
\begin{proof}[Proof of Theorem \ref{theorem: DP-ANNC Unbalanced}]  Let $T$ be the counter from Algorithm \ref{alg: reduction}, $\tilde{T}$ its differential private version, and  $I(\bff{q})$ be the set of indices of the buckets the algorithms needs to inspect, then $\widehat{\text{ans}} = \sum_{i \in I(\bff{q})} T[i]$ and $\widetilde{\text{ans}} = \sum_{i \in I(\bff{q})} \tilde{T}[i]$. The application of Laplace noise leads to an unbiased and uncorrelated estimator $\tilde{T}[i]$ so the variance of the error is
\begin{align*}
    \text{Var}[\widetilde{\text{ans}}-\widehat{\text{ans}}] &= \mathbb{E}\bigg[\bigg(\sum_{i\in I({\bf q})}(\tilde{T}[i]-T[i])\bigg)^2\bigg]
    =\mathbb{E}[I(\bff{q})]\cdot \text{Var}[\text{Lap}(1/\varepsilon)] \leq  O\bigg(\frac{\mathcal{K}^2}{\varepsilon^2}\bigg),
\end{align*}
as $\tilde{T}[i] = T[i]+Z$ where $Z\sim \text{Lap}(1/\varepsilon)$ and each noise is sampled independently. Therefore, by Jensen's inequality $\mathbb{E}[|\widetilde{\text{ans}}-\widehat{\text{ans}}|] \leq  \sqrt{\text{Var}[\widetilde{\text{ans}}-\widehat{\text{ans}}]} \leq O\big(\frac{\mathcal{K}}{\varepsilon}\big)$ and then by Markov's inequality $|\widetilde{\text{ans}}-\widehat{\text{ans}}|\leq O\big(\frac{\mathcal{K}}{\varepsilon}\big)$ holds with constant probability. The claim follows by Lemma \ref{lemma: Reduction from ANN to ANNC} and $\varepsilon\leq 1$. The additional $O(m)$ space and pre-preprocessing time is necessary to store and sample $m$ i.i.d. independent Laplace random variables, one for each element of the partition $\mathcal{L}=(L_1,\dots, L_m)$.
\end{proof}
\subsection{Omitted Proofs in Section \ref{section: Analysis}}
\begin{lemma}
\label{lemma: collision}
The probability that a point $\bff{x}\in \mathbb{S}^{d-1}$ collides during \texttt{CloseTop-1} construction is at least $1-\frac{1}{m^{\Omega(1)}}$.
\end{lemma}
For the proof of Lemma \ref{lemma: collision} we first need the following technical lemma.
\begin{lemma}
\label{lemma: techincal lemma} Let $m$ be any integer greater than $4$. Define $a = \sqrt{2\log m}-\frac{3}{2}\frac{\log\log m}{\sqrt{2\log m}}$ and $b= \sqrt{2\log m}$, then for $Z\sim \mathcal{N}(0,1)$, we have $\Pr[Z \in (a,b)] \geq \frac{2\sqrt{\pi}}{3}\frac{\log m}{m}$.
\end{lemma}
\begin{proof}
     Using Proposition \ref{proposition: loose tail gaussian} we may bound $\Pr[Z\geq b] \leq \frac{4\sqrt{\pi}}{3}\frac{1}{m\sqrt{\log m}}$. For the left side of the interval, we first need to check if $a \geq 1$. We have that $\sqrt{2\log m}-\frac{3}{2}\frac{\log\log m}{\sqrt{2\log m}}\geq 1$ only if $3 \leq \frac{4\log m -2\sqrt{2\log m}}{\log\log m}$. But $\frac{4\log m -2\sqrt{2\log m}}{\log\log m}> 5$ for any $m\geq 5$. Thus, by applying Proposition \ref{proposition: loose tail gaussian} we get 
     \begin{align*}
         \Pr[Z\geq a] & \geq \frac{2\sqrt{2\pi}}{(1+\sqrt{5})}\frac{1}{\sqrt{2\log m}-\frac{3}{2}\frac{\log\log m}{2\sqrt{2\log m}}}\exp\left[-\frac{1}{2}\left(\sqrt{2\log m}-\frac{3}{2}\frac{\log\log m}{\sqrt{2\log m}}\right)^2\right]\\
         & \geq \frac{2\sqrt{\pi}}{(1+\sqrt 5)}\frac{1}{m\sqrt{\log m}}\log^{3/2}m\exp\left[-\frac{9}{16}\frac{(\log\log m)^2}{\log m}\right] \geq \frac{4\sqrt{\pi}}{3}\frac{\log m}{m}.
     \end{align*}
     Where the last inequality holds if $\frac{9}{16}\frac{(\log\log m)^2}{\log m}\leq \log(2(1+\sqrt{5})/3)$. The right-hand side is greater than $1/2$, thus, it is sufficient to check $\frac{(\log\log m)^2}{\log m}\leq 8/9$. For $m\geq 5$ the left-hand side is always smaller\footnote{The maximum is reached at $m=11$.} than $1/3$, thus, the inequality is satisfied.
     Putting these two bounds together, we conclude
    \begin{align*}
        \Pr[Z\in (a,b)] = \Pr[Z\geq a]-\Pr[Z\geq b] &\geq \frac{4\sqrt{\pi}}{3}\frac{\log m}{m}\left(1-\frac{1}{(\log m)^{3/2}}\right) \geq \frac{2\sqrt{\pi}}{3}\frac{\log m}{m}.
    \end{align*}
    The last inequality follows from $(\log m)^{3/2}\geq (\log 5)^{3/2} \geq 2$.
\end{proof}
\begin{proof}[Proof of Lemma \ref{lemma: collision}]
    If the probability that a random vector succeeds in the assignation is $p$, then a point will not collide with probability $(1-p)^m$. Then for $p = \Omega\big(\frac{\log m}{m}\big)$ (from Proposition \ref{lemma: techincal lemma}) the probability to not collide is at most $(1-p)^m \leq e^{-pm} = \frac{1}{m^{\Omega(1)}}$.
\end{proof}
\begin{proof}[Proof of Lemma \ref{lemma: CloseTop-1 lemma 1}]
The probability to not find a close point $\bff{x}_{\alpha}$ is 
\begin{align}
    \Pr[Q_{\bff{x}_\alpha}\leq \eta]&=\underset{Z\sim \mathcal{N}(0, 1-\alpha^2)}{\Pr}[Z\leq \eta -\alpha X_{\bff{x}_\alpha}]\notag\\
    &\leq \underset{Z\sim \mathcal{N}(0, 1-\alpha^2)}{\Pr}\left[Z \leq -\sqrt{2(1-\alpha^2)\log\log m}\left(1-\frac{3}{4}\sqrt{\frac{\alpha^2}{1-\alpha^2}\frac{\log\log m}{\log m}}\right)\right]\notag\\
    &\leq \Pr\left[Z \leq -\sqrt{2(1-\alpha^2)\log\log m}\left(1-O\left(\sqrt{\frac{\log\log n}{\log n}}\right)\right)\right]\notag\\
    &\leq \log m^{-1} (\log m)^{O(\sqrt{\log \log n/\log n})}\leq O(\log^{-1}m)
    \label{equation: probability close point},
\end{align}
where in the first equality we used Lemma \ref{lemma: relation between concominant and extreme}, in the second inequality we use the fact that $X_{\bff{x}}\geq \sqrt{2\log m}-\frac{3}{2}\frac{\log\log m}{\sqrt{2\log m}}$ by construction, in the third inequality $\frac{\log\log m}{\log m} = O\big((1-\alpha^2)\frac{\log\log n}{\log n}\big)$ for $m = n^{\frac{\theta}{1-\alpha^2}}$, and lastly $\lim_{n\to \infty}(\log n)^{O(\sqrt{\log\log n/\log n})} = 1$. The probability to find a close point is the probability of the joint event $[Q_{\bff{x}_\alpha}\geq \eta]$ and ${\bf x}_{\alpha}$ is stored in the data structure. Thus, by Lemma~\ref{lemma: collision}, we have that $\Pr[\text{find } \bff{x}_{\alpha}]\geq (1-O(\log^{-1}m))(1-m^{-\Omega(1)}) = 1-o(1)$.
This proves Lemma \ref{lemma: close point probability} for \texttt{CloseTop-1}. The probability to inspect a far point $\bff{x}_\beta$ is 
\begin{align}
    \Pr[Q_{\bff{x}_{\beta}}\geq \eta] & =\underset{Z\sim \mathcal{N}(0, 1-\beta^2)}{\Pr}[Z\geq \eta -\beta X_{\bff{x}_\beta}]\notag\\
    &\leq \underset{Z\sim \mathcal{N}(0, 1-\beta^2)}{\Pr}[Z\geq  (\alpha-\beta)\sqrt{2\log m}-\sqrt{2(1-\alpha^2)\log\log m}]\label{eq: far point close top 1}
\end{align}
where in the first equality we used Lemma \ref{lemma: relation between concominant and extreme}, while in the following inequality we used the fact that $X_{\bff{x}}\leq \sqrt{2\log m}$ by construction. The analysis then follows the same step of Lemma \ref{lemma: Expected number of far points}. As the analysis of the expected number of buckets to inspect is the same, Lemma \ref{lemma: Expected number of far points} holds under the same assumption.
\end{proof}

\begin{proof}[Proof of Lemma \ref{lemma: close point tensorclosetop1}]
    Let's consider one data structure $\mathcal{D}_i$, due to Equation \ref{equation: probability close point} we have an upper bound of $O\big(\frac{t}{\log m}\big)$ to not find a close point in $\mathcal{D}_i$. By applying a union bound over $t$ data structures we have that
\begin{equation*}
    \Pr\bigg[\bigcup_{i=1}^{t}\{\text{not find $\bff{x}_{\alpha}$ in $\mathcal{D}_i$}\}\bigg]\leq O\bigg(\frac{t^2}{\log m}\bigg)=O\bigg(\frac{1}{(1-\alpha^2)\log^{3/4}n}\bigg)=o(1),
\end{equation*}
where we used $m=n^{\frac{\theta}{1-\alpha^2}}$, and $(1-\alpha^2) = \omega(\log^{-3/4}n)$. We now study the probability to not store a point. Due to Lemma \ref{lemma: collision} the probability to not store a point in $\mathcal{D}_i$ is $\frac{1}{m^{\Omega(1/t)}}$, then by a union bound we have
\begin{equation*}
\label{equation: remove point probability}
      \Pr\bigg[\bigcup_{i=1}^{t}\{\text{$\bff{x}$ is not stored in  $\mathcal{D}_i$}\}\bigg]\leq \frac{t}{m^{\Omega(1/t)}}=o(1)\underbrace{\log^{7/8}n\cdot e^{-\Omega(\log^{7/8}n)}}_{=o(1)}=o(1),
\end{equation*}
as $t=o(\log^{7/8}n)$ for $1-\alpha^2 = \omega(\log^{-3/4}n)$, and $m^{\Omega(1/t)}=n^{\Omega(\log^{-1/8}n)} = e^{\Omega(\log^{7/8}n)}$ for $m=n^{\frac{\theta}{1-\alpha^2}}$. Therefore, a close point is found with at least $1-o(1)$ probability.
\end{proof}

\begin{proof}[Proof of Lemma \ref{lemma: far points tensorclos}]
    Consider one \texttt{CloseTop-1} data structure $\mathcal{D}_i$ with $\tilde{m} = n^{\frac{1}{t}\frac{\theta}{1-\alpha^2}} = n^{\frac{\theta}{\log^{1/8}n}}$ random vectors. Thus $\frac{\log\log \tilde{m}}{\log \tilde{m}} = O\big(\frac{\log\log n}{\log^{7/8}n}\big)$, so that the threshold may be written as $\eta \geq \alpha\sqrt{2\log\tilde{m}}\big(1-\frac{\sqrt{1-\alpha^2}}{\alpha}O\big(\sqrt{\frac{\log\log n}{\log^{7/8}n}}\big)\big)$, which is positive for $\alpha\geq \alpha-\beta = \Omega\big(\sqrt{\frac{\log\log n}{\log^{7/8}n}}\big)$. Therefore, by following the same computation of Lemma~\ref{lemma: Expected number of far points} (Equation~\ref{eq: probability check vector}), the expected number of buckets to inspect in $\mathcal{D}_i$ is at most $\tilde{m}^{1-\alpha^2 + \sqrt{1-\alpha^2}o(1)}=n^{\frac{1}{t}(\theta + \frac{1}{\sqrt{1-\alpha^2}}o(1))}$. By assumption, we have $1/\sqrt{1-\alpha^2} = o(\sqrt{\log^{3/4}n})$, thus $\frac{1}{\sqrt{1-\alpha^2}}O(\sqrt{\frac{\log\log n}{\log^{7/8}n}})=o(\sqrt{\frac{\log\log n}{\log^{1/8}n}}) = o(1)$. By tensorization of $t$ independent data structures, we conclude that the expected number of buckets is at most $n^{\theta + o(1)}$.

    Analogously, starting from the computation in Lemma~\ref{lemma: CloseTop-1 lemma 1} (Equation~\ref{eq: far point close top 1}) and substituting $m$ with $\tilde{m}$, we may lower bound the threshold with $(\alpha-\beta)\sqrt{2\log\tilde{m}}\big(1-\sqrt{\frac{1-\alpha^2}{(\alpha-\beta)^2}}O\big(\sqrt{\frac{\log\log n}{\log^{7/8} n}}\big)\big)$, which is positive for $(\alpha-\beta)=\Omega\big(\sqrt{\frac{\log\log n}{\log^{7/8}n}}\big)$. Thus, by following the computation in Lemma~\ref{lemma: Expected number of far points} (Equation~\ref{eq: far point}), the probability to find a far point is at most $\tilde{m}^{-\frac{(\alpha-\beta)^2}{1-\beta^2}+\frac{\sqrt{1-\alpha^2}(\alpha-\beta)}{1-\beta^2}o(1)}$. The probability to find a far point in all the $t$ independent data structures is at most
    \begin{equation*}
        \Pr\left[\bigcap_{i=1}^{t}\{{\bf x}_{\beta} \text{ is found in } \mathcal{D}_i \}\right]\leq \tilde{m}^{t\left(-\frac{(\alpha-\beta)^2}{1-\beta^2}+\frac{\sqrt{1-\alpha^2}(\alpha-\beta)}{1-\beta^2}o(1)\right)} = n^{-\theta\frac{(\alpha-\beta)^2}{(1-\alpha^2)(1-\beta^2)}+\frac{(\alpha-\beta)}{\sqrt{(1-\alpha^2)}(1-\beta^2)}o(1)}.
    \end{equation*}
    The last addend is still $o(1)$ as $\frac{1}{\sqrt{1-\alpha^2}}O(\sqrt{\frac{\log\log n}{\log^{7/8}n}})=o(1)$ and $\frac{\alpha-\beta}{1-\beta^2} = O(1)$. Thus, as there are at most $n$ far points, the expected number of far points that are inspected is $n^{1-\theta\frac{(\alpha-\beta)^2}{(1-\alpha^2)(1-\beta^2)}}$.
\end{proof}

\section{Data Structures for the Euclidean Space}
\label{appendix: Embedding into the Euclidean sphere}
In this section, we prove that balanced \texttt{TensorCloseTop-1} solves the $(c,r)$-ANN in the Euclidean space and reproduces the results for $(c,r)$-ANNC in \cite{andoni2024differentially}. Due to standard embedding techniques (see Lemma A.1 and Corollary A.1 in \cite{andoni2024differentially}),  a $(c,r)$-ANN in $\mathbb{R}^{d}$ can be mapped into a $(c\cdot \frac{1-\gamma}{1+\gamma}, r(1+\gamma))$-ANN in $\mathbb{S}^{d'}$, in time $O(d\cdot d')$, with $d' = O\big(\frac{\log n}{\gamma^2}\big)$, if $(cr)^2 \leq \gamma/2$. Thus, the embedding preserves asymptotically the metric only for small distances $r=o(1)$ in $\mathbb{S}^{d'-1}$, that can be obtained in the original space $\mathbb{R}^{d}$ after an appropriate scaling. The relation between inner product similarity and Euclidean distance for small distances are
\begin{equation}
\label{equation: relation with r}
    (1-\alpha^2) = \Theta(r^2), \qquad (1-\beta)^2 = \Theta(r^2), \qquad (\alpha-\beta) = \Theta(r^2), \qquad (1-\alpha\beta) = \Theta(r^2),
\end{equation}
and the concatenation factor $t = \frac{\log^{1/8}n}{1-\alpha^2}$ is in $\Theta\big(\frac{\log^{1/8}n}{r^2}\big)$.

\begin{theorem}
\label{theorem: embedding}
For any $r>0$, constant $c>1$, and a dataset $\mathcal{S}=\{x_i\}_{i=1,\dots, n}$ in $\mathbb{R}^{d}$ , there exists a data structure that solves with at least $1-o(1)$ probability the $(c,r)$-ANN using almost linear space $n^{1+o(1)}$, pre-processing time $d\cdot n^{1+o(1)}$, and query time in expectation at most $d\cdot n^{o(1)} + n^{\rho + o(1)}$ for $\rho = \frac{4c^2}{(c^2+1)^2}$.
\end{theorem}
\begin{proof}
To apply \texttt{TensorCloseTop-1} to the embedded dataset we need to satisfy the assumptions in Theorem \ref{theorem: TensorCloseTop-1} which are: (i) $\rho = \frac{(1-\alpha^2)(1-\beta^2)}{(1-\alpha\beta)^2} = O(1)$, (ii)  $\alpha-\beta = \Omega\bigg(\sqrt{\frac{\log\log n}{\log^{7/8} n}}\bigg)$, and (iii) $1-\alpha^2 = \omega(\log^{-3/4}n)$. Requirement (i) is satisfied due to the asymptotic Equations \ref{equation: relation with r}. More precisely, by substituting $\alpha = 1-\frac{r^2}{2}$ and $\beta = 1-\frac{(cr)^2}{2}$ we get
\begin{equation*}
    \rho = \frac{(1-\alpha^2)(1-\beta^2)}{(1-\alpha\beta)^2} = \frac{c^2(4-r^2)(c^2r^2-4)}{(c^2(r^2-2)-2)^2} = \frac{4c^2}{(c^2+1)^2} + O(r^2).
\end{equation*}
Requirement (ii) and (iii) are satisfied for any distance $r=\omega(\log^{-3/8}n)$,  due to the asymptotic relations in Equation \ref{equation: relation with r}. Therefore, for any $C<3/8$ by setting $\gamma = \log^{-2C}n$ and $r=\log^{-C}n$ we can scale the dataset $\mathcal{S}\subset \mathbb{R}^{d}$, apply the standard embedding techniques to get a dataset in $\mathbb{S}^{d'-1}$ with $d'= \log^{O(1)}n$, and invoke \texttt{TensorCloseTop-1} to solve the $(\alpha,\beta)$-ANN in $\mathbb{S}^{d'}$ by paying an asymptotically small $\gamma$ factor.\footnote{Andoni et al. \cite{andoni2024differentially} set $r = \Theta(\log^{-1/8}n)$ and $\gamma=\Theta(\log^{-1/8}n)$. Our analysis demonstrates more clearly that there is a broader range of possible values.} 
As the mapping can be computed in $O(d\cdot d')= O\big(d\cdot (\log n)^{O(1)}\big)=d\cdot n^{o(1)}$ time, the pre-processing time is $d\cdot n^{1+o(1)}$. 
The space is $O(d'\cdot (n+n^{o(1)})) = n^{1+o(1)}$ and the query time is in expectation at most $d\cdot n^{o(1)} + n^{\rho + o(1)}$, given by the time to embed the query $d\cdot n^{o(1)}$ in the hyper-sphere plus the query time of the data structure $d'\cdot n^{\rho + o(1)} = n^{\rho + o(1)}$.
\end{proof}
\paragraph*{The Unbalanced \texttt{TensorCloseTop-1}} Unbalanced \texttt{TensorCloseTop-1} can be used to solve the Euclidean DP-ANNC problem as well. The proof if the same of Theorem \ref{theorem: embedding}, with the distinction that
\begin{equation*}
    \frac{\sigma}{2} = \frac{(1-\alpha^2)(1-\beta^2)}{(1-\alpha\beta)^2+ (1-\alpha\beta)} = \frac{2c^2}{1+c^4}+O(r^2)
\end{equation*}
The space and pre-processing time needed is $n^{\frac{2\sigma}{1-\alpha^2}} = n^{\frac{2\sigma}{\Theta(r^2)}} = n^{\text{polylog}(n)}$.

\end{document}